\documentclass[journal]{IEEEtran}
\usepackage{cite}
\usepackage{amsmath,amssymb,amsfonts}
\usepackage{algorithmic}
\usepackage{graphicx}
\usepackage{textcomp}
\usepackage{xcolor}
\usepackage{doi}
\usepackage{amsthm}
\usepackage{setspace}
\usepackage{physics}
\usepackage{algorithm}
\usepackage{algorithmic}
\usepackage{mathtools}
\usepackage{diagbox}
\usepackage{caption}
\usepackage{subcaption}
\usepackage{physics}
\usepackage{multirow}
\usepackage{tabularx}
\usepackage{booktabs}
\usepackage[labelfont=bf,format=plain,justification=raggedright,singlelinecheck=false]{caption}
\newtheorem{definition}{Definition}
\newtheorem{theorem}{Theorem}
\newtheorem{remark}{Remark}
\newtheorem{lemma}{Lemma}

\newtheorem{proposition}{Proposition}

\renewcommand{\algorithmiccomment}[1]{\bgroup\hfill//~#1\egroup}
\newcommand{\etal}{\textit{et al. }}
\begin{document}
\title{Distributed Task Management in Fog Computing:
A Socially Concave Bandit Game}
\author{Xiaotong Cheng and Setareh Maghsudi
%
\thanks{X. Cheng is with the Department of Computer Science, University of Tübingen, 72074 Tübingen, Germany (email:xiaotong.cheng@uni-tuebingen.de). S. Maghsudi is with the Department of Computer Science, University of Tübingen, 72074 Tübingen, Germany and with the Fraunhofer Heinrich Hertz Institute, 10587 Berlin, Germany. A part of this paper appeared in \cite{cheng2022distributed} at the 2022 IEEE Signal Processing Workshop for Wireless Communications (SPAWC). This work was supported by Grant MA 7111/6-1 from the German Research Foundation (DFG) and Grant 16KISK035 from the German Federal Ministry of Education and Research (BMBF)}
}
\maketitle
\begin{abstract}
Fog computing leverages the task offloading capabilities at the network's edge to improve efficiency and enable swift responses to application demands. However, the design of task allocation strategies in a fog computing network is still challenging because of the heterogeneity of fog nodes and uncertainties in system dynamics. We formulate the distributed task allocation problem as a social-concave game with bandit feedback and show that the game has a unique Nash equilibrium, which is implementable using no-regret learning strategies (regret with sublinear growth). We then develop two no-regret online decision-making strategies. One strategy, namely bandit gradient ascent with momentum, is an online convex optimization algorithm with bandit feedback. The other strategy, Lipschitz bandit with initialization, is an EXP3 multi-armed bandit algorithm. We establish regret bounds for both strategies and analyze their convergence characteristics. Moreover, we compare the proposed strategies with an allocation strategy named learning with linear rewards. Theoretical- and numerical analysis shows the superior performance of the proposed strategies for efficient task allocation compared to the state-of-the-art methods.
\end{abstract}
\IEEEpeerreviewmaketitle
\section{Introduction}
\label{Sec:Intro}
\IEEEPARstart{I}{n} the past years, the world has witnessed the explosive demand for excessively resource-consuming wireless services and applications such as online gaming and on-demand streaming. Cloud computing is envisioned as a promising paradigm to handle the intensive computation demands and high energy consumption because it virtually enables an unlimited storage capacity and also processing power \cite{diaz2016state} and has dynamic characteristics such as elasticity, reduced management efforts, and flexible pricing model (pay-per-use) \cite{mouradian2017comprehensive}. However, cloud computing suffers from some shortcomings, the most fundamental ones being delay and excessive bandwidth consumption, both caused by the long (physical or logical) distance between the centralized cloud servers and devices \cite{bittencourt2017mobility}. In addition, the impact of the data centers' power consumption on the environment is a significant concern \cite{hu2015mobile,vaquero2014finding}. 

The \textit{fog computing} paradigm addresses the challenges by extending the cloud architecture towards the network edge. Indeed, it brings the computation in the proximity of end-users by introducing a hierarchy of computing capacities, often decentralized devices, referred to as fog nodes and cloudlets between the edge and the cloud \cite{xiao2018distributed}. The architecture exploits the idle communication and computational resources on the demand side to provide low-latency computing, storage, and network services. Besides, a recent study shows that fog computing is a promising solution to reduce energy consumption on both the mobile device side and cloud server side \cite{chang2017energy}.

To realize the potential of fog computing, efficient allocation of the offloaded tasks to the fog nodes plays a crucial role. Developing a sophisticated online task allocation strategy is, however, challenging. One significant issue is the heterogeneity of fog nodes and tasks. Indeed, in the future ultra-dense networks with ubiquitous connectivity, users offload various tasks, including transmission, sensing, and signal processing. Besides, several powerful end-users, edge nodes, and access devices act as fog nodes. Such nodes are distinct concerning software and hardware; thus, they have different characteristics and capabilities and perform offloaded tasks at various quality and efficiency levels. The challenge becomes aggravated in a distributed structure when avoiding information exchange to retain low signaling and feedback overhead. In addition, another challenge is related to the trade-off between individual interest and social welfare: the fog nodes might be reluctant to share their resources for extra tasks to avoid expensive computation and energy loss unless they receive appropriate compensation \cite{zhou2019resource}. Furthermore, the energy efficiency is also a crucial factor in developing the task allocation algorithm \cite{xiao2018distributed}.

Based on the discussion above, it is imperative to develop task allocation policies that are amenable to distributed implementation and robust to information shortage and uncertainty in the fog computing paradigm. In this work, we address all the above challenges. Specifically, we investigate the task allocation problem in a distributed fog computing architecture. Our contributions are summarized below: 
\begin{itemize}
\item Our system model is generic as we allow for arbitrary heterogeneous fog nodes and quantify their distinct characteristics using a well-designed utility function, which can be regarded as a measure of the quality of experience (QoE) based on task completion efficiency and resources (including power) consumption. Besides, our proposed framework is more realistic as fog nodes do not need any back-and-forth negotiation or disclose information about each other's types, the tasks' utility, and the cost.
\item Taking the selfishness and rationality of the intelligent fog nodes into account, we model the task allocation problem as a sequential decision-making game. We prove that the task allocation game is a social-concave game, which converges to the Nash equilibrium when every player uses a no-regret learning strategy to select tasks. Besides, we prove that the Nash equilibrium in the formulated task allocation game is unique.
\item To deal with the uncertainties in system dynamics and limited information, we reformulate the problem as a game with bandit feedback. Based on \cite{flaxman2004online} and \cite{bubeck2011lipschitz}, we propose two decision-making policies: Bandit gradient ascent with momentum (BGAM) and Lipschitz bandit with initialization (LBWI), and prove their performance characteristics such as regret bound and convergence. 
\item Through intensive numerical analysis, we evaluate the performance of our scheme in comparison to an existing centralized approach and several distributed benchmark methods working based on various principles.
\end{itemize}

The rest of the paper is organized as follows. \textbf{Section~\ref{sec:rel-work}} discusses the related work. \textbf{Section~\ref{sec:sysm}} presents the system model and basic assumptions. In \textbf{Section~\ref{sec:pf}}, we formulate the problem of task sharing among heterogeneous entities under uncertainty and model the formulate task allocation problem as a game. We also discuss the existence and uniqueness of Nash equilibrium. In \textbf{Section~\ref{sec:nrs}}, we develop and analyze two no-regret decision-making strategies that converge to Nash equilibrium: The first one uses the bandit gradient descent (BGD) (\textbf{Section~\ref{sec:bgam}}), while the second one is based on Lipschitz bandit (LB) (\textbf{Section~\ref{sec:lb}}). \textbf{Section~\ref{sec:llr}} describes our benchmark strategy for performance evaluation and comparison. \textbf{Section~\ref{sec:num-ana}} includes numerical analysis and the subsequent discussions. \textbf{Section~\ref{sec:concl}} concludes the paper and suggests some directions for future research.
\section{Related Work} \label{sec:rel-work}
Some existing works solve the allocation problems using centralized optimization methods. \textit{Deng \etal} develop a systematic framework to solve a workload allocation problem in cloud-fog computing. The allocation problem is decomposed into several subproblems \cite{deng2015towards}: (i) workload allocation in fog and cloud servers based on the power consumption and computation delay, and (ii) minimizing communication delay in dispatch between fog and cloud servers. In \cite{chen2016joint}, the offloading decisions and the allocation of communication resources are jointly optimized via separable semidefinite relaxation. Then the work is extended in \cite{chen2018multi} in a complicated scenario where a computing access point (CAP) serving as the network gateway and a computation service provider, is connected with the extended mobile cloud computing. The solution proposed in \cite{chen2016joint} is generalized with an added CAP with additional alternating optimization and sequential tuning.

In the concept of fog computing, intelligent user devices with extra computational resources act voluntarily as fog nodes while simultaneously being concerned with their interests and privacy. Therefore, one cannot force them to share all information and perform all tasks based on the instruction of the centralized controller \cite{zhou2019resource,peng2017data}. Recently, a large body of research has resorted to game-theoretical methods for distributed allocation of generalized resources. \textit{Kaewpuang \etal} propose a decision-making framework in a mobile cloud computing environment, where resource allocation and revenue management are realized via cooperation formation among mobile service providers \cite{kaewpuang2013framework}. They model the cooperation formation among providers as a distributed game and use the best-response dynamics as a solution. A similar algorithm concerning cooperation under uncertainty appears in \cite{Mittal21:DCU}. The method applies to several applications, including fog computing. In \cite{chiti2018matching}, the authors formulate the task offloading problem as a matching game with externalities. A strategy based on the deferred acceptance algorithm is proposed, which enables efficient allocation in a distributed mode and ensures stability over the matching outcome. References \cite{chen2015efficient} and \cite{shah2018hierarchical} model the problem as a potential game that admits a pure strategy Nash equilibrium. The former uses the finite improvement property to achieve Nash equilibrium, whereas the latter proposes a best-response adaption algorithm to achieve Nash equilibrium. Besides, \cite{shah2018hierarchical} extends their framework and proposes a near-optimal solution for resource allocation to address the long convergence time to achieve equilibrium. In \cite{Ranadheera18:COA}, the authors model the task offloading problem as a minority game and obtain the Nash equilibrium. Reference \cite{Habiba19:ARA} uses a reverse auction model for resource allocation and task management in a computation offloading platform as a part of a software-defined network. In \cite{maghsudi2018distributed}, the authors study a task allocation problem among heterogeneous cyber-physical systems under state uncertainty. They use the concept of deterministic equivalence and sequential core to solve the problem and use the Walrasian auction process to implement the core. Besides, in \cite{Maghsudi21:CHD}, the authors propose a supply-demand market model for task allocation, where the fog nodes and offloading user, respectively, represent the sellers that set the service prices and the consumer that uses the services at the given prices. However, most game-theoretical strategies require the information exchange \cite{kaewpuang2013framework, Mittal21:DCU,shah2018hierarchical} among fog nodes about their decisions, and some even allow the back-and-forth negotiation, which does not take the fog nodes' privacy secure into account.

Alongside game theory, online learning is another widely-used mathematical tool to efficiently share or allocate resources under uncertainty. In \cite{gai2012combinatorial}, the authors propose a multi-armed bandit setting to solve the combinatorial user-channel allocation problem and a novel policy learning called Learning with Linear Rewards (LLR) is proposed. Its main step is solving a deterministic combinatorial optimization with a linear objective. Reference \cite{Ghoorchian21:MBE} develops a decision-making method, namely, budget-limited sliding window UCB, based on multi-armed bandit theory to select the most suitable server in a dynamic time-variant network. Reference \cite{liao2019robust} designs a two-stage task offloading approach. In the first stage, the algorithm produces a contract that specifies the contribution and associated reward to encourage the fog servers to share resources. The second stage is a version of the upper confidence bound (UCB) method to connect the user and fog server. In \cite{zhou2019resource}, the authors develop a similar two-stage resource sharing and task offloading approach, which considers not only task offloading under information asymmetry \cite{liao2019robust} and also the scenario without information asymmetry. Besides, the performance of the proposed distance-aware, occurrence-aware, and task-property-aware volatile upper confidence bound algorithm is rigorously analyzed in terms of theoretical regret bound in \cite{zhou2019resource}. In general, the multi-armed bandit is a suitable framework to solve the allocation problem with discrete action space, whereas it is inefficient for the continuous action space. Some other related work considers solving the allocation problem with continuous action space with optimization strategies or machine learning strategies \cite{chen2020joint,wang2020machine}, in which the gradient information is necessary, which is also challenging in practical applications.

Unlike previous works, we merge game theory and online bandit learning into a unified framework to solve the task allocation problem in a distributed manner with limited information. Distributed learning with bandit feedback guarantees the fog nodes' privacy, security, and low communication cost. 
\section{System Model}
\label{sec:sysm}
We consider a fog computing system consisting of $K$ fog nodes (FNs) gathered in the set $\mathcal{K}=\{1,2,\dots,K\}$. Besides, there is a set of $M$ offloaded tasks denoted by $\mathcal{M} = \{1,2,\dots,M\}$. \textbf{Fig.~\ref{fig:rs}} shows an instant of such fog computing system.
\begin{figure}[!htp]
  \vspace{-10pt}
  \centering
  \includegraphics[width=0.6\linewidth]{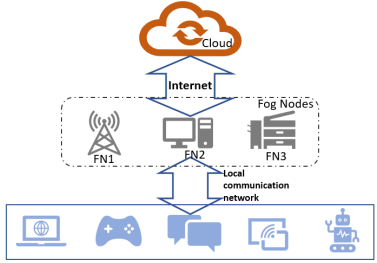}
  \caption{Distributed task allocation among the fog nodes.}
  \label{fig:rs}
\end{figure}

Each fog node decides to which task(s) and at which level it contributes. Naturally, there might remain some tasks that are not appealing to fog nodes (e.g., the task with an excessive energy consumption); Such tasks will be uploaded to the cloud server. Since fog nodes may correspond to the same type of devices and any computational task can be associated with many fog nodes, we assume that the fog nodes can cooperatively accomplish tasks. This assumption is widely used in task allocation in fog computing systems \cite{xiao2018distributed,wang2019learning}. Vector $\boldsymbol{x}_{k}=(x_{k,1}, \dots, x_{k,M})$, $0 \leq x_{k,m}\leq 1$, gathers the fraction of each task $m \in \mathcal{M}$ that fog node $k \in \mathcal{K}$ selects to perform. Fog nodes are different concerning ability and capacity. We quantify such a type heterogeneity using an efficiency index vector $\boldsymbol{\rho}_{k} = (\rho_{k,1}, \dots,\rho_{k,M})$ for every fog node $k \in \mathcal{K}$. $\boldsymbol{\epsilon}_k = (\epsilon_{k,1}, \ldots, \epsilon_{k,M})$ contains the power consumption index vector. Besides, we characterize the cost, e.g., the memory cost, using a cost index vector $\boldsymbol{\kappa}_{k}=(\kappa_{k,1}, \dots,\kappa_{k,M})$. We assume $\forall k \in \mathcal{K}, m \in \mathcal{M}$, $0 < \rho_{\min} \leq \rho_{k,m} \leq \rho_{\max}$, $0 < \epsilon_{\min} \leq \epsilon_{k,m} \leq \epsilon_{\max}$ and $0 < \kappa_{\min} \leq \kappa_{k,m} \leq \kappa_{\max}$.

After each fog node $k\in \mathcal{K}$ announces its preferred task share $\boldsymbol{x}_{k}$, the tasks are allocated among all fog nodes to satisfy their request as far as possible. For the fog node $k$, we denote the allocation vector by $\boldsymbol{a}_{k}=(a_{k,1}, \dots, a_{k,M})$, where
\begin{gather}
a_{k,m}=\frac{x_{k,m}}{\sum_{i \in \mathcal{K}}x_{i,m}}, \label{eq:akm}
\end{gather}
if $\sum_{i \in \mathcal{K}}x_{i,m}>0$, and $a_{k,m} = 0$ otherwise.\footnote{Alternatively, one can use the \textit{barrier to entry} concept, i.e., add a small constant to the denominator \cite{Heliou20:GFO}.} The mechanism described by (\ref{eq:akm}) corresponds to a proportional allocation mechanism, which is an intuitive way to allocate tasks/resource fairly \cite{nguyen2019market}. Under the proportional allocation, every fog node receives a fraction of the task equivalent to its requested proportion divided by the sum of all fog nodes requested proportion. 

We assume that every fog node $k\in \mathcal{K}$ in the multi-agent system acts as a risk-averse player since it is a popular assumption in multi-agent systems \cite{maghsudi2017distributed,ranadheera2017minority}. Therefore, we model the task completion efficiency reward by using a risk-aversion model as \footnote{Some other functions can also be used as the reward function as long as some regularity conditions are satisfied.}
\begin{gather}
\label{eq:UtilityOne}
\varphi_{k,m}(a_{k,m})=\rho_{k,m}(1-e^{-\frac{1}{\rho_{k,m}}a_{k,m}}). 
\end{gather}
By (\ref{eq:UtilityOne}), for a fixed amount $a_{k,m}$, a larger efficiency index $\rho_{k,m}$ results in higher reward. Thus, the fog node attempts to perform more share of tasks in which it is more efficient (higher efficiency index) so that the system efficiency increases \cite{maghsudi2018distributed}. In real applications, the heterogeneity coefficient follows from the computational capability of the devices and the transmission cost to the cloud server. Besides, after submitting its preferred task share, each fog node pre-allocates (reserves) the required resources to handle its requested task share, including power. Concerning the energy efficiency, we consider the power consumption efficiency of fog nodes, similar as \cite{xiao2018distributed,cheng2016statistical}. Let $\epsilon_{k,m}x_{k,m}$ denote the power reservation of fog node $k$ for task $m$. We define the power consumption efficiency as the required fraction of $\epsilon_{k,m}x_{k,m}$ to process the allocated share of task $m$. Formally,
\begin{gather}
e_{k,m} = \frac{\epsilon_{k,m}x_{k,m}}{a_{k,m}}.
\end{gather} 
Similarly, for a fog node $k \in \mathcal{K}$, the cost of resource reservation is proportional to the requested task fraction. Formally,
\begin{gather}
c_{k,m}(x_{k,m})=\kappa_{k,m}x_{k,m}. 
\end{gather}
Therefore, the total utility of the fog node $k \in \mathcal{K}$ yields
\begin{gather}
u_{k}(\boldsymbol{x}_{k},\mathbf{X}_{-k})=\sum_{m \in \mathcal{M}}\rho_{k,m}(1-e^{-\frac{a_{k,m}}{\rho_{k,m}}}) + \frac{\epsilon_{k,m}}{a_{k,m}}- \kappa_{k,m}x_{k,m},
\label{eq:util}
\end{gather} 
where $\boldsymbol{X}_{-k} = (\boldsymbol{x}_1, \dots, \boldsymbol{x}_{k-1}, \boldsymbol{x}_{k+1}, \dots, \boldsymbol{x}_K)$ denotes the joint decision of all fog nodes excluding $k$. To ensure the information privacy, safety, and low communication cost, each fog node can only access its utility information after performing tasks, which can be regarded as a bandit feedback. 

In the following, we state two propositions concerning the utility function defined in (\ref{eq:util}).
\begin{proposition} 
\label{prop:lips}
For every $k \in \mathcal{K}$, if $\sum_{i \in \mathcal{K}}x_{i,m}>0$, the utility function defined in (\ref{eq:util}) satisfies the Lipschitz condition 
\begin{gather*}
|u_{k}(\boldsymbol{x})-u_{k}(\boldsymbol{y})|\leq L \cdot |\boldsymbol{x}-\boldsymbol{y}|, \quad \text{for any}\  \boldsymbol{x}, \boldsymbol{y} \in [0,1]^{M}.
\end{gather*}
\end{proposition}
\begin{proof}
See Appendix \ref{Sec:Proposition1}.
\end{proof}
\begin{proposition} 
\label{prop:conc-conv}
For every $k \in \mathcal{K}$, if $\sum_{i \in \mathcal{K}}x_{i,m}>0$, the utility function defined in (\ref{eq:util}) is concave in the proposal of fog node $k$, $\boldsymbol{x}_k$, and convex in the actions of other fog nodes, $\boldsymbol{X}_{-k}$.  
\end{proposition}
\begin{proof}
See Appendix \ref{Sec:Proposition2}.
\end{proof}
\section{Problem Formulation} 
\label{sec:pf}
The task management procedure follows in successive rounds $t=1,2,\dots,T$; to show that explicitly, we add the superscript $t$ to the notation. Similar to many previous work such as \cite{maghsudi2014joint,xu2013opportunistic}, we use the zero-average Gaussian noise to model the random factors and uncertainties in the fog computing system and add a noise variable to the utility of each fog node. Formally, at each round $t$, by proposing $\boldsymbol{x}_{k}^{t}$, the total utility of fog node $k \in \mathcal{K}$ yields $u_{k}^{t}(\boldsymbol{x}_{k}^{t},\mathbf{X}_{-k}^{t})+N_{0}$, where $N_{0}$ is a zero-average Gaussian noise; that is, the fog node observes a noisy version of the average utility resulted from $(\boldsymbol{x}_{k}^{t},\mathbf{X}_{-k}^{t})$ in (\ref{eq:util}). 

Primarily, each selfish fog node $k \in \mathcal{K}$ opts to maximize its accumulated utility
\begin{align}
\underset{\boldsymbol{x}_k \in [0,1]^M}{\textup{maximize}}
\sum_{t=1}^T \sum_{m=1}^M & \rho_{k,m}(1 - e^{-\frac{x_{k,m}^t}{\rho_{k,m}\sum_{i \in \mathcal{K}}x_{i,m}^t}}) \notag \\
& + \epsilon_{k,m}\sum_{i \in \mathcal{K}}x_{i,m}^t- \kappa_{k,m}x_{k,m}^t  \label{eq:pro-form}.
\end{align} 

Solving optimization problem \eqref{eq:pro-form} is extremely difficult as (i) the fog nodes do not have any prior information about the tasks and the utility functions, (ii) the utility of each fog node depends on the action of the other nodes, and (iii) after each round of decision-making, each fog node only observes the utility of the performed action and receives no other feedback.

The scenario described above corresponds to a \textit{bandit setting} \cite{Maghsudi16:MAB}, where the agents adjust their actions based on the estimation from feedback information. In this setting, the performance metric follows the conventional concept of expected total regret. We define the cumulative regret of fog node $k$ up to time $T$ as
\begin{gather}
R_{k}^{T} = \sum_{t=1}^{T} \mathbb{E}[\max_{\boldsymbol{x}_k} u_{k}^{t}(\boldsymbol{x}_{k},\boldsymbol{X}_{-k}^{t})- u_{k}^{t}(\boldsymbol{x}_{k}^{t},\boldsymbol{X}_{-k}^{t})].
\label{eq:Reg}
\end{gather} 
Then, instead of solving (\ref{eq:pro-form}), each fog node minimizes its accumulated regret $R_{k}^{T}$. Regret minimization procedure enables every agent to have a guarantee on its utility regardless of the actions of others \cite{even2009convergence}.
%
\subsection{Task Allocation Game} 
\label{sec:tag}
From a distributed system perspective, the fog nodes must achieve a steady-state. By regarding fog nodes as the players and their task selection as action, we solve the problem as a socially concave game with bandit feedback.

We define the described task allocation game $\Gamma$ by a tuple $\prec \mathcal{K},\{\mathcal{S}_k\}_{k\in \mathcal{K}}, \{\boldsymbol{x}_{k}\}_{k\in \mathcal{K}}, \{\varphi_{k}\}_{k\in \mathcal{K}}, \{u_{k}\}_{k\in \mathcal{K}}\succ$, where: (i) $\mathcal{K}$ is the set of players, implying that each fog node is a player; (ii) $\mathcal{S}_{k} \in [0,1]^M$ is player $k$'s strategy space; (iii) $\boldsymbol{x}_{k} \in [0,1]^M$ is player $k$'s action; (iv) $\varphi_{k}: [0,1]^M \rightarrow \mathbb{R}$ is the value function of player $k$ assigning a value/reward function to allocation; (v) $u_{k}(\boldsymbol{x}_{k},\boldsymbol{X}_{-k}): \otimes_{k \in \mathcal{K}}\mathcal{S}_{k} \rightarrow \mathbb{R}$ is the utility function of player $k$, defined as $u_{k}(\boldsymbol{x}_{k},\boldsymbol{X}_{-k}) = \varphi_{k}(\boldsymbol{a}_{k}(\boldsymbol{x}_{k},\boldsymbol{X}_{-k})) + \epsilon_{k}\boldsymbol{X}_{-k} - \kappa_{k} \boldsymbol{x}_{k}$. In the formulated task allocation game, each fog node (player) selects some tasks to perform in consecutive rounds to minimize the cumulative regret defined in \eqref{eq:Reg}. Before proceeding to analyze the task allocation game, we present some notions and concepts.
\begin{definition}[Nash Equilibrium \cite{even2009convergence}] 
\label{def:NashEq}
A strategy profile $\boldsymbol{X}^{*} = (\boldsymbol{x}_{1}^{*}, \ldots, \boldsymbol{x}_{k}^{*}, \ldots, \boldsymbol{x}_{K}^{*})$ is called Nash equilibrium if for all $k \in \mathcal{K}$ and all actions $\boldsymbol{x}'_{k}$,  strategy profile $\boldsymbol{X} = (\boldsymbol{x}_{1}^{*}, \ldots, \boldsymbol{x}'_{k}, \ldots, \boldsymbol{x}_{K}^{*})$ yields
\begin{gather}
u_k(\boldsymbol{x}_{k}^{*},\boldsymbol{X}_{-k}^{*}) \geq u_k(\boldsymbol{x}'_{k},\boldsymbol{X}_{-k}^{*}).
\end{gather}
A strategy profile $\boldsymbol{X}^{*}$ is called $\epsilon$-Nash equilibrium if for all $k \in \mathcal{K}$ and all actions $\boldsymbol{x}'_{k}$, strategy profile $\boldsymbol{X}$ yields
\begin{gather}
u_k(\boldsymbol{x}_{k}^{*},\boldsymbol{X}_{-k}^{*}) \geq u_k(\boldsymbol{x}'_{k},\boldsymbol{X}_{-k}^{*}) - \epsilon.
\end{gather}
\end{definition}
\begin{definition}[Concave Game \cite{Ui98:CCG}] 
\label{def:concv}
A game $\Gamma$ is concave if for all $k \in \mathcal{K}$, $u_k(\boldsymbol{x}_k)$ is a concave function in $\boldsymbol{x}_k$ for every fixed $\boldsymbol{X}_{-k} \in \mathcal{S}_{-k}$.  
\end{definition}
\begin{definition}[Socially Concave Game \cite{even2009convergence}]
\label{def:sc-concv}
A game $\Gamma$ is socially concave if 
\begin{itemize}
\item[(i)]  $\exists~\{\lambda_{k}\}_{k \in \mathcal{K}}$, $\lambda_{k} > 0$, $\sum_{k \in \mathcal{K}}\lambda_k=1$, such that the weighted sum of utility functions $\sigma(\boldsymbol{X},\boldsymbol{\lambda}) = \sum_{k=1}^K \lambda_k u_k(\boldsymbol{x}_k)$ is a concave function in every $\boldsymbol{x}_{k}$, where $\boldsymbol{X} = (\boldsymbol{x}_{1}, \ldots, \boldsymbol{x}_{k}, \ldots, \boldsymbol{x}_{K})$ is the joint strategy profile.
\item[(ii)]  The utility function of each player $k \in \mathcal{K}$ is convex in the actions of the other players, i.e., for every $\boldsymbol{x}_{k} \in \mathcal{S}_{k}$, the function $u_{k}(\boldsymbol{x}_{k},\boldsymbol{X}_{-k})$ is convex in $\boldsymbol{X}_{-k} \in \mathcal{S}_{-k}$.
\end{itemize}
\end{definition}
The following propositions describe some properties of the task allocation game $\Gamma$ concerning equilibrium.
\begin{proposition} 
\label{prop:ta-sc}
The task allocation game $\Gamma$ by a tuple $\prec \mathcal{K},\{\mathcal{S}_k\}_{k\in \mathcal{K}}, \{\boldsymbol{x}_{k}\}_{k\in \mathcal{K}}, \{\varphi_{k}\}_{k\in \mathcal{K}}, \{u_{k}\}_{k\in \mathcal{K}}\succ$ is a socially concave game and a concave game.
\end{proposition}
\begin{proof}
See Appendix \ref{Sec:Proposition3}.
\end{proof}
%
\begin{proposition}
\label{prop:ta-ue}
If the task allocation game $\Gamma$ converges to a Nash equilibrium, then that equilibrium is unique. 
\end{proposition}
\begin{proof}
See Appendix \ref{Sec:Proposition4}.
\end{proof}
\section{Efficient Equilibrium Implementation via No-Regret Bandit Strategy} \label{sec:nrs}
In the task allocation procedure, each fog node $k \in \mathcal{K}$ aims at minimizing its regret. Furthermore, as discussed in \textbf{Section \ref{sec:tag}}, from a system perspective, the fog nodes' interactions must converge to a steady-state or equilibrium. To solve the formulated game, we propose two bandit learning strategies based on the \textit{bandit gradient descent} (BGD) algorithm \cite{flaxman2004online} and \textit{Lipschitz Bandit} (LB) algorithm \cite{bubeck2011lipschitz}. In the following, we describe the developed decision-making policies and establish a regret bound. More precisely, we prove that both proposed strategies are no-regret, meaning that they guarantee sub-linear regret growth.
\textbf{Table~\ref{tab:nota}} summarizes the frequently-used notations of this section.

Specially, by (\ref{eq:util}), the utility of every fog node $k \in \mathcal{K}$ is additive over all tasks $m \in \mathcal{M}$. Besides, the fog node selects the fraction of tasks to perform independently. Therefore, the multi-task allocation problem boils down to $M$ independent single-task allocation problems with the same set of participants. In the following, we consider the single agent single task selection strategy. In numerical implementations, the strategy of each agent runs in parallel.
\begin{table}[!ht]
\begin{center}
\centering
\captionsetup{justification=centering}
\caption{Notation}
\label{tab:nota}
 \begin{tabular}{||c|p{5.7cm}||}
 \hline
 \multicolumn{2}{||c||}{BGAM for node $k$ and task $m$} \\ 
 \hline
 Notation & Meaning  \\ [0.5ex] 
 \hline
 $y^t$ & Approximation of shifted action \\
 $z^t$ & Shifted action \\
 $u^t$ & Utility function with shifted action \\
 $x_{k,m}^t$ & Action of node $k$ in task $m$  \\
 \hline
 \multicolumn{2}{||c||}{LBWI for node $k$ and task $m$} \\ 
 \hline
 Notation & Meaning  \\ [0.5ex] 
 \hline
 $N$ & The number of intervals in Phase I  \\ 
 $\hat{\mu}^t \in \mathbb{R}^N$ & The average utility matrix \\ 
 $\mathcal{A}^t \in \mathbb{R}^N$ & The matrix for counting  \\
 $\mathcal{A}$ & Times performed in each interval in Phase I\\
 $\Omega^t \in \mathbb{R}^N$ & The weight matrix in Phase I \\
 $T_1$ & Number of rounds in Phase I, $T_1 = \mathcal{A}N$ \\
 $\tilde{L}$ & The estimation of Lipschitz constant  \\
 $\tilde{N}$ & The number of intervals in Phase II\\ 
 $\omega^t \in \mathbb{R}^N$ & The weight matrix in Phase II \\
 $T$ & Number of rounds in Phase II \\[1ex] 
 \hline
 \end{tabular}
 \end{center}
 \vspace{-20pt}
\end{table}
\subsection{Bandit Gradient Ascent with Momentum}
\label{sec:bgam}
The BGD strategy \cite{flaxman2004online} solves the convex optimization problem with bandit feedback based on the algorithm proposed by Zinkevich \cite{zinkevich2003online} via the one-point approximation of the gradient. Our proposed method differs from the BGD strategy in two crucial aspects: (i) we adapt it to gradient ascent, and (ii) we add momentum in the update rule to accelerate the convergence. 

To better adjust to BGD policy, we shift the action space $[0,1]$ to $[-\xi,\xi]$, with $\xi$ being a positive constant. Formally, let $y^t$ be the running parameter and $z^{t}$ denote the shifted action 
\begin{gather}
z^{t} = y^{t} + \sigma c^t, 
\label{eq:z^t}
\end{gather}
with $c^t$ being a unit vector selected uniformly at random. As $z^{t}$ is the shifted action, we have
\begin{gather}
x_{k,m}^{t}=z^{t}+\xi. 
\label{eq:action}  
\end{gather}
Besides, we define the function $u^{t}(z^{t})$, which is the utility function with the shifted action space; that is, $u^{t}(z^{t})=u_{k,m}^{t}(z^{t}+\xi)= u_{k,m}^{t}(x_{k,m}^{t})$. Let \cite{flaxman2004online} 
\begin{gather}
\hat{u}^{t}(y^{t}) = \mathbb{E}_{c^t \in \mathbb{B}}[u^{t}(y^{t}+\sigma c^{t})] = \mathbb{E}_{c^t \in \mathbb{B}}[u^{t}(z^{t})], 
\label{eq:est-z}
\end{gather}
where $\sigma$ is a positive constant, and $\mathbb{B} = \{x \in \mathcal{R}||x| \leq 1\}$ is a unit ball centered around the origin in action space. Thus $y^t$ can be interpreted as the approximation of the shifted action. 
\begin{lemma}[\cite{flaxman2004online}]
\label{lem:apxg} 
Fix $\sigma > 0$, over the random unit vectors c,
\begin{gather}
\mathbb{E}[u(y+\sigma c)c] = \sigma\nabla \hat{u}(y).
\end{gather}
\end{lemma}
\noindent According to \textbf{Lemma~\ref{lem:apxg}}, in the algorithm BGAM, we can use
\begin{gather}
\sigma \nabla \hat{u}^{t}(y^{t})= u^{t}(z^{t})c^{t} =  u_{k,m}^{t}(x_{k,m}^{t})c^{t} 
\label{eq:egrad}
\end{gather}
to estimate the gradient.

As mentioned previously, we further enhance the algorithm with momentum, which is a trick in gradient-based optimization algorithms. It allows gradient-based optimizers to speed up along low curvature directions. Here, we apply the momentum with the one-point gradient estimation that provably achieves the same enhancement as in gradient algorithms. Hence the update rule follows as
\begin{align}
v^t=\beta v^{t-1}+\sigma\nabla \hat{u}^t(y^t), \label{eq:v} \\
y^{t+1} = P_{(1-\alpha)S}(y^t + \nu^t v^t),
\label{eq:y}
\end{align}
where $\alpha$ and $\beta$ are constant numbers. Besides, $P(\cdot)$ is the projection function $P_{S}(y) = \min_{x \in S} \norm{x-y}$. The reason to obtain $y^{t+1}$ in $(1-\alpha)S$ is to avoid the value of $z^t$ out of the defined action range \cite{flaxman2004online}. \textbf{Algorithm~\ref{alg:bgam}} summarizes the proposed BGAM decision-making strategy.
\begin{algorithm}[!htp]
\caption{Bandit Gradient Ascent with Momentum \\ for fog node $k \in \mathcal{K}$ and task $m \in \mathcal{M}$ }
\label{alg:bgam}
\begin{algorithmic}[0]
\STATE \textbf{Parameters}: $\xi=0.5$
\STATE \textbf{Initialization}: $y^{1} = 0, v^{1} = 0$.
\end{algorithmic}
\begin{algorithmic}[1]
\FOR{$t =1,2,3, \dots, T$}
\STATE Select unit vector $c_t$ uniformly at random;
\STATE Calculate $z^{t}$ using (\ref{eq:z^t});
\STATE Take action $x_{k,m}^t$ using (\ref{eq:action});
\STATE Estimate the gradient using (\ref{eq:egrad});
\STATE Update $v^{t}$ using (\ref{eq:v});
\STATE Obtain $y^{t+1}$ using (\ref{eq:y});
\ENDFOR
\end{algorithmic}
\end{algorithm}

\subsubsection{Complexity}
BGAM algorithm only stores the updated parameters. As such, the space requirement is $\mathcal{O}(1)$. The total runtime at $T$ is $\mathcal{O}(T)$. Besides, in the distributed structure, the communication overhead is not costly with bandit feedback, each fog node only shares its preferred task proportion information and receives the feedback information until the convergence of equilibrium. Thus the communication overhead is at most $\mathcal{O}(KT)$.
\subsubsection{Regret Analysis}
In the following, we establish a regret-bound for the BGAM algorithm. Before proceeding to the theory, it is essential to mention that introducing a momentum in BGAM does not influence its convergence. Rather, it accelerates the convergence rate in practical implementation. 
\begin{proposition}
\label{pr:RegretBGAM}
Let $u^{1},u^{2},\cdots,u^{T}: \mathcal{S} \rightarrow \mathbb{R}$ be a sequence of concave and differentiable functions. Besides, $g^{1},g^{2},\cdots,g^{T}$ are the single-point estimation of gradient with $g^{t} = \nabla \hat{u}^{t}(y^{t})$ and $\norm{g^{t}} \leq G$. Let $L$ be the Lipschitz constant of utility function. Assume that each $y^{t}$ generated by Algorithm~\ref{alg:bgam} satisfies $\norm{y^{i}-y^{j}}_{2} \leq D$, for all $i,j \in \{1,2,\cdots,T\}$ and $D$ being a positive constant. Select $\nu^{t}=\frac{\nu}{\sqrt{t}}$, $\sigma = T^{-0.25}\sqrt{\frac{RUr}{3(Lr+U)}}$, and $\alpha = \frac{\sigma}{r}$, where $U$ is the maximum value of utility function. Also, $r$ and $R$ satisfy the condition that shifted strategy space $S$ contains the ball of radius $r$ centered at the origin and is contained in the ball of radius $R$ \cite{flaxman2004online}. Here, we select $r=R=\xi$. Then the decision-making policy BGAM guarantees the following regret bound:
\begin{align}
\mathbb{E}[R(T)] \leq  \tilde{\mathcal{O}}(T^{\frac{3}{4}})
\end{align}
\end{proposition}
\begin{proof}
See Appendix \ref{Sec:Proposition5}.
\end{proof}
\subsection{Lipschitz Bandit with Initialization} 
\label{sec:lb}
In the seminal MAB problem, the set of arms is finite. In many applications, that setting only serves as an imprecise model of the situation, which potentially forces the learner to pull only suboptimal arms, thereby causing a linear regret \cite{trovo2016budgeted}. In contrast, the Continuum-Armed Bandit (CAB) defines the set of arms over a continuous space. Such a setting has attracted intensive attention in the past few years due to its ability to model general situations. In particular, it fits the problems where the expected reward is a Lipschitz function of the arm, known as \textit{Lipschitz Bandits}.
\begin{proposition}
\label{pro:Lip}
For fog node $k$, the task allocation game can be modelled as a Lipschitz multi-armed bandit problem.
\end{proposition}
\begin{proof}
See Appendix \ref{Sec:Proposition6}.
\end{proof}
A simple approach to solving the MAB problems with continuous arm space is discretizing the action space; nevertheless, determining the optimal quantization intervals is challenging and has a remarkable impact on the regret bound. Reference \cite{bubeck2011lipschitz} proposes a decision-making policy based on discretization, which consists of two phases: In the first phase, the method explores uniformly to find a crude estimate of the Lipschitz constant and determines the optimal number of intervals for discretization. In the second phase, it finds the best one using a standard exploration-exploitation strategy.

The BGAM algorithm proposed in Section~\ref{sec:bgam} only requires the bandit feedback; nevertheless, the information about the Lipschitz constant is necessary to optimize the hyperparameter according to \textbf{Proposition~\ref{pr:RegretBGAM}}. However, the Lipschitz bandit algorithm proposed in \cite{bubeck2011lipschitz} does not impose such a limitation, as the available approximation of the Lipschitz constant suffices for optimization. Besides, it guarantees a sublinear regret growth. 
\subsubsection{Lipschitz Bandit with Initialization (LBWI)} 
Based on the properties of the task allocation game, here we adopt  \textit{Exponential-weight algorithm for Exploration and Exploitation} (EXP3), whose core idea is to assign some selection probability to each discretized action at every trial, which is proportional to the exponentially weighted accumulated reward of that action \cite{auer2002nonstochastic}.

Reference \cite{bubeck2011lipschitz} shows that the optimal number of discretization intervals depends directly on the number of rounds T, i.e., a longer game requires more discretization intervals. However, using too many expands the exploration rate dramatically, thus reducing the efficiency and wasting computational resources. Therefore, we modify the classical Lipschitz bandit strategy by adding several steps between the two phases. More precisely, we implement an initialization process for the second phase with the information collected during the first phase. The initialization helps the algorithm effectively use the prior information about the action selection collected at the pure exploration phase. Therefore, it accelerates convergence despite increasing the memory cost.

In the first phase of the LBWI strategy, the agent quantizes the action space into $N$ intervals (arms), where $N$ is a random positive integer. Afterward, it performs pure exploration by pulling each arm  $\mathcal{A}$ times. We use $\lambda^t$ to denote the selected arm at time $t$, where $\lambda^t \in\{0,1,\ldots,N-1\}$ is an integer to indicate the selected discretized interval. To model a general situation with continuous action space, the real action $x_{k,m}^t$ is sampled from the corresponding interval as
\begin{gather}
x_{k,m}^t = \text{uniform\_sample}(\frac{\lambda^t}{N},\frac{\lambda^t+1}{N}). 
\label{eq:action-lbwi}
\end{gather}
Let $\hat{\mu}^t \in \mathbb{R}^N$ be a vector that collects the average utility for each specific arm. Besides, $\mathcal{A}^t \in \mathbb{R}^N$ stores the number of times each specific arm has been selected,
\begin{align}
\hat{\mu}^t[n] &= \begin{cases} 
\text{$\frac{\hat{\mu}^{t-1}[n]\mathcal{A}^{t-1}[n]+u_{k,m}^t}{\mathcal{A}^{t-1}[n]+1}$,}&\text{if $\lambda^t = n$,} \\
\text{$\hat{\mu}^{t-1}[n]$,}&\text{otherwise,}   
\end{cases} \qquad \\
%
\mathcal{A}^t[n] &= \begin{cases} 
\text{$1+\mathcal{A}^{t-1}[n]$,}&\quad\text{if $\lambda^t = n$,}\\
\text{$\mathcal{A}^{t-1}[n]$,}&\quad\text{otherwise.}
\end{cases} 
\label{eq:avg_num}
\end{align}
At the same time, in the first phase, a weighted vector $\Omega \in \mathbb{R}^N$ and a probability vector $p \in \mathbb{R}^N$ related to the utility are recorded. The update rule of the weighted matrix here is the same as that in EXP3 strategy. The first phase provides an approximation of the Lipschitz constant as
\begin{align}
\hat{L} &= N \max_{n \in \{0,1,\cdots, N-1\}}  \max_{ i\in \{-1,1\}} |\hat{\mu}^{T_1}[n] - \hat{\mu}^{T_1}[n+i]|, \label{eq:hatL} \\
\tilde{L} &= \hat{L} + N\sqrt{\frac{2}{\mathcal{A}}\ln(2NT)} 
\label{eq:tildeL}.
\end{align}
Then the EXP3 strategy in Phase II is performed based on the estimated Lipschitz constant obtained from Phase I. The number of discretized action intervals in Phase II is
\begin{gather}
\tilde{N} = N \left \lceil{\frac{\tilde{L}^{\frac{2}{3}}T^{\frac{1}{3}}}{N}}\right \rceil.
\end{gather}
As described before, besides estimating the Lipschitz constant and the optimal number of discretization intervals, we use the weight matrix $\Omega^{T_1}$ to feed some prior information into Phase II; Nonetheless, as the number of intervals changes in Phase II, it is essential to redistribute the weights $\Omega^{T_1}$ in Phase I to initialize the weights $\omega^{T_1+1}$ in Phase II. With the approximated Lipschitz constant and initialized weight matrix, the EXP3 strategy finds the optimal discretized action interval for solving the task allocation problem.

\textbf{Algorithm~\ref{alg:lbwi}} summarizes the proposed LBWI decision-making strategy. 
\begin{algorithm}[!htp]
\caption{Lipschitz Bandit with Initialization Strategy-\\
Pure Exploration Phase}
\label{alg:lbwi}
\begin{algorithmic}[0]
\STATE \textbf{Parameters}: \\
$N$: number of intervals for discretization in Phase I; \\
$\gamma$: Real constant number $\gamma \in (0,1]$;
\STATE \textbf{Initialization}: \\
$\Omega$: Weight matrix with $\Omega^1[n] = 1, n = 0,\ldots, N-1$;
\STATE \textbf{Pure Exploration Phase (Phase I)}
\end{algorithmic}
\begin{algorithmic}[1]
\FOR{$t = 1,2,3, \cdots, T_1$}
\STATE Select the arm $\lambda^t$ randomly;
\STATE Sample action $x_{k,m}^t$ randomly according to (\ref{eq:action-lbwi}) and get the utility $u_{k,m}^t$;
\STATE Update the average utility vector $\hat{\mu}^t$ and action selection matrix $\mathcal{A}^t$ according to (\ref{eq:avg_num});
\STATE Set $p^t[n] = (1-\gamma)\frac{\Omega^t[n]}{\sum_{i=1}^N\Omega^t[i]} + \frac{\gamma}{N}, n = 0,1,\cdots,N-1$;
\STATE Set 
\begin{gather}
\Omega^{t+1}[n] = \begin{cases} 
\text{$\Omega^t[n]\exp{\frac{\gamma u_{k,m}^t}{Np^t[n]}}$,}&\text{if $\lambda^t = n$}\\
\text{$\Omega^t[n]$,}&\text{otherwise}
\end{cases} \notag
\end{gather}
\ENDFOR
\STATE Obtain $\hat{L}$ using (\ref{eq:hatL}) and $\tilde{L}$ using (\ref{eq:tildeL}).
\end{algorithmic}
\begin{algorithmic}[0]
\STATE \textbf{EXP3-Phase (Phase II)}
\end{algorithmic}
\begin{algorithmic}[1]
\FOR{$n = 0, 1,\cdots,\tilde{N}-1$}     
\STATE $i = \lceil \frac{n}{N} \rceil$, $\omega^{T_1+1}[n] = \frac{N}{\tilde{N}}\Omega^{T_1}[i]$; \COMMENT{Initialization}
\ENDFOR
\FOR{$t = T_1+1, \cdots, T$}
\STATE  Set $p^t[n] = (1-\gamma)\frac{\omega^t[n]}{\sum_{i=1}^{\tilde{N}}\omega^t[i]} + \frac{\gamma}{\tilde{N}}, n =0,1,\ldots,\tilde{N}-1$;
\STATE Select arm $\lambda^t$ according to probability matrix $\boldsymbol{p}^t$;
\STATE Sample the corresponding action $\boldsymbol{x}_{k,m}^t$ according to $\lambda^t$
\begin{gather}
x_{k,m}^t = \text{uniform\_sample}(\frac{\lambda^t}{\tilde{N}},\frac{\lambda^t+1}{\tilde{N}});
\end{gather}
\STATE Receive the reward/utility $u_{k,m}^t$ and update the weight matrix as
\begin{gather}
\omega^{t+1}[n] = \begin{cases} 
\text{$\omega^t[n]\exp{\frac{\gamma u_{k,m}^t}{\tilde{N}p^t[n]}}$,}&\text{if $\lambda^t = n$}\\
\text{$\omega^t[n]$,}&\text{otherwise}
\end{cases} \notag
\end{gather}
\ENDFOR
\end{algorithmic}
\end{algorithm}
\subsubsection{Complexity}
The complexity of the two-phase Lipschitz bandit algorithm is much more than that of BGAM. In the exploration phase, the runtime is $\mathcal{O}(T_1), T_1 << T$, whereas the runtime of the EXP3 phase is $\mathcal{O}(T^{\frac{4}{3}})$ because of the round-related discretization number $\tilde{N}$. Therefore, its total runtime yields $\mathcal{O}(T^{\frac{4}{3}})$. The total space requirement is $\mathcal{O}(T^{\frac{1}{3}})$. The communication overhead is the same as BGAM, thus it is upper bounded by $\mathcal{O}(KT)$.
\subsubsection{Regret Analysis}
In the following, we establish a regret-bound for the LBWI algorithm. 
\begin{proposition}
\label{pr:RegretLBWI}
Let $L$ be the Lipschitz constant of utility function and let $H$ be the uniform bound of utility function's Hessians. Select $N$ for division in phase I satisfying $N \geq \frac{8H}{L}$. With probability at lease $1-\frac{1}{T}$, the decision-making policy LBWI achieves the expected regret bound: 
\begin{gather}
\mathbb{E}[R(T)] \leq \tilde{\mathcal{O}}(T^{\frac{5}{6}})
\end{gather}
\end{proposition}
\begin{proof}
See Appendix \ref{Sec:Proposition7}.
\end{proof}
\begin{proposition}
\label{pr:Convergence}
Consider the task allocation game formulated in \textbf{Section \ref{sec:tag}}. If all of the fog nodes implement BGAM or LBWI decision-making strategies, then the game converges to the unique Nash equilibrium. 
\end{proposition}
\begin{proof}
See Appendix \ref{Sec:Proposition8}.
\end{proof}
\begin{remark}
The proposed strategies, i.e., BGAM and LBWI, are generic in the sense that they are implementable not only for the utility function in \eqref{eq:util} but also for any other utility function that satisfies the Lipschitz- and the concave-convexity conditions given by \textbf{Proposition}~\ref{prop:lips} and \ref{prop:conc-conv}. In particular, for BGAM, the Lipschitz condition is not a necessity.
\end{remark}
\section{Learning with Linear Rewards Strategy} 
\label{sec:llr}
Learning with linear rewards (LLR) strategy \cite{gai2012combinatorial} is a decision-making policy for the stochastic combinatorial multi-armed bandit problem with linear rewards, i.e., when the total reward is a linearly-weighted combination of the selected random variables. The policy combines the upper confidence bound (UCB) and combinatorial optimization strategies. We briefly review this strategy and use it later as a benchmark for numerical performance evaluation.

Different with proposed strategies, LLR strategy has more limitations: every node must select one task at each round, meaning $x_{k,m} \in \{0,1\}$ and $\sum_{m \in \mathcal{M}} x_{k,m} = 1$. Therefore, if $M > K$, some tasks remain unassigned. To store the information after taking an action at a time slot, the policy uses two metrics, namely sample mean matrix $\hat{\theta}_{K \times M}$ and observed times matrix $\mathcal{C}$, defined as
\begin{align}
\hat{\theta}_{k,m}^t &=  \begin{cases} 
\text{$\frac{\hat{\theta}_{k,m}^{t-1}\cdot\mathcal{C}_{k,m}^{t-1}+u_{k,m}^t}{\mathcal{C}_{k,m}^{t-1}+1}$,}&\text{if $x_{k,m}^t = 1$,} \\
\text{$\hat{\theta}_{k,m}^{t-1}$,}&\text{otherwise,}   
\end{cases}  \qquad \\ 
\mathcal{C}_{k,m}^t &= \begin{cases} 
\text{$1+\mathcal{C}_{k,m}^{t-1}$,}&\quad\text{if $x_{k,m}^t = 1$,}\\
\text{$\mathcal{C}_{k,m}^{t-1}$,}&\quad\text{otherwise.}
\end{cases} \label{eq:avg_n}
\end{align}
For initialization, the LLR policy runs a random play to ensure that every action is played at least once ($\mathcal{C}_{k,m} > 0$). Afterward, unlike distributed strategies, it uses the maximum weighted matching to allocate the tasks in a centralized way. More precisely, we consider the bipartite graph $\mathcal{G}=(V,E)$, where $V = \{\mathcal{K},\mathcal{M}\}$ is the vertex set that includes the fog nodes and the tasks. The LLR strategy assigns one task to each agent via the classical bipartite matching algorithm, such as the Hungarian algorithm. The edge weight $(k,m)$ represents the upper confidence bound of the utility of the fog node (agent) $k$ in performing task $m$, which is updated as follows
\begin{gather}
\mathcal{W}_{k,m}^t = \hat{\theta}_{k,m}^t + \sqrt{\frac{(\Xi+1)\ln t}{\mathcal{C}_{k,m}^t}} \label{eq:weights},
\end{gather}
where $\Xi = \text{min}\{K,M\}$. \textbf{Algorithm~\ref{alg:llr}} summarizes the LLR strategy.
\begin{algorithm}[!htp]
\caption{Learning with Linear Rewards}
\label{alg:llr}
\begin{algorithmic}[1]
\FOR{$t = 1,2,3, \cdots, T_1$}
\STATE Select the arm $x_{k,m}^t$ randomly;
\STATE Update the average utility vector $\hat{\theta}_{k,m}$ and action selection matrix $\mathcal{C}_{k,m}$ according to (\ref{eq:avg_n});
\ENDFOR
\FOR{$t = T_1+1,\ldots, T$}
\STATE Play an action which solves the maximum weight matching problem with weights matrix $W$ according to (\ref{eq:weights})
\STATE Update the average utility vector $\hat{\theta}_{k,m}$ and action selection matrix $\mathcal{C}_{k,m}$ according to (\ref{eq:avg_n}).
\ENDFOR
\end{algorithmic}
\end{algorithm}
\subsubsection{Analysis} 
Unlike BGAM and LBWI, LLR is a centralized policy. It requires all fog nodes to share their information with a central coordinator, thus increasing the overhead and computational cost. Besides, it allocates each agent only one task, which is a less flexible allocation mechanism ignoring the different capabilities of fog nodes. That aspect, despite reducing the computational cost slightly, the remaining unprocessed tasks will lead to a linear increase in regret over time. Concerning the privacy aspects and the selfishness of the fog nodes, BGAM and LBWI are more efficient than LLR. We establish this claim also through numerical analysis in the next section.
\section{Numerical Analysis} 
\label{sec:num-ana}
We divide the numerical analysis into two parts. In \textbf{Section~\ref{sec:na-g1}}, we consider a toy scenario with two fog nodes and two tasks. The goal is to clarify the workflow of the developed task allocation schemes. Then, in \textbf{Section~\ref{sec:na-g2}}, we expand the scenario significantly by increasing the number of tasks and fog nodes (FNs). The goal is to analyze the performance of the proposed strategies compared with the state-of-art solutions.
\subsection{Game I}
\label{sec:na-g1}
\subsubsection{Model Parameter}
In the two-server two-tasks allocation game, we show the two fog nodes' action profiles as $\begin{bmatrix}
[x_{11},x_{12}] \\
[x_{21},x_{22}] \\
\end{bmatrix}$.  The efficiency index $\boldsymbol{\rho} = \begin{bmatrix}
0.9 & 0.5 \\
0.6 & 0.85 \\
\end{bmatrix}$, power consumption index $\boldsymbol{\epsilon} = \begin{bmatrix}
0.1 & 0.03 \\
0.05 & 0.2 \\
\end{bmatrix}$ and the resource consumption index $\boldsymbol{\kappa}= \begin{bmatrix}
0.1 & 0.8 \\
0.75 & 0.05 \\
\end{bmatrix}$, where $\rho_{k,m}, \kappa_{k,m}, k,m \in\{1,2\}$ respectively correspond to the performance index and basic energy consumption of node $k$ in task $m$. 

\textbf{Fig.~\ref{fig:up22}} shows the 3D plots of each node's utility function. It indicates that the utility functions are Lipschitz continuous. Intuitively, if FN2 refuses to participate in Task 1, the utility of both fog nodes becomes maximum. Similarly, if FN1 does not contribute to Task 2, both nodes' utility is maximum, which is consistent with the mathematical proofs as well as the intuition of decision-making: When the power efficiency index is relative small, if the cost consumption index is higher than the performance index for some specific task, then the fog server prefers to avoid performing that task. For example, for $\rho_{1,2} \leq \kappa_{1,2}$, then FN1 hesitates to contribute to Task 2.
\begin{figure}
  \vspace{-15pt}
  \centering
  \includegraphics[width=.8\linewidth]{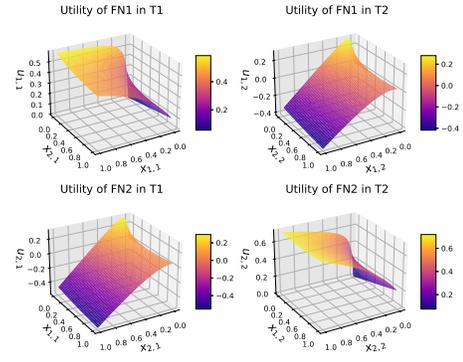}
  \vspace{-15pt}
  \captionof{figure}{3D plots of the utility of two FNs.}
  \label{fig:up22}
  \vspace{-10pt}
\end{figure}%
\begin{figure*}
\centering
\begin{minipage}{.35\textwidth}
  \centering
  \hspace{-15pt}
  \includegraphics[width=.95\linewidth]{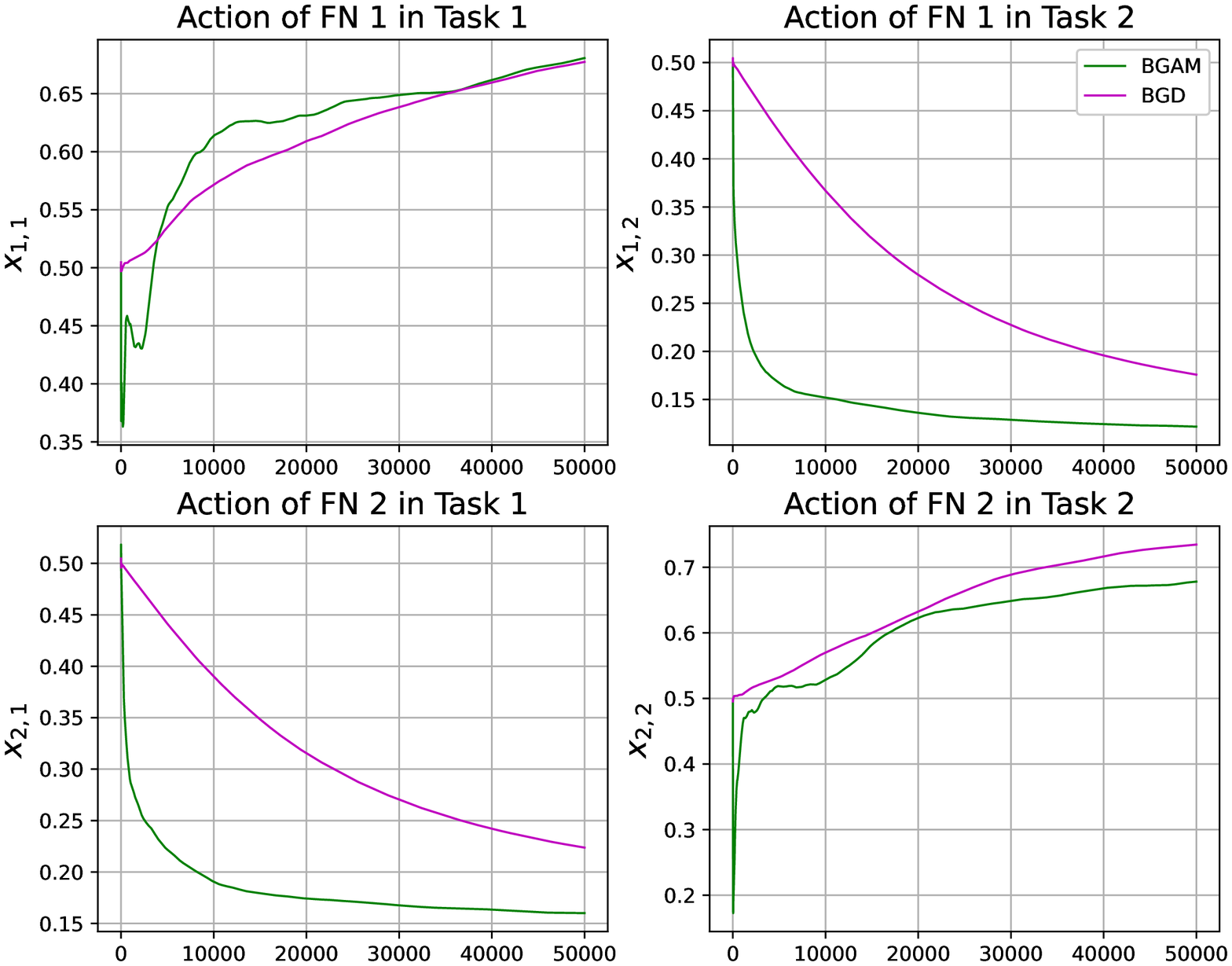}
  \vspace{-5pt}
  \captionof{figure}{Actions of two FNs that use the bandit gradient algorithms.}
  \label{fig:fnact1}
\end{minipage}
\begin{minipage}{.35\textwidth}
\begin{subfigure}{.98\textwidth}
  \centering
  \includegraphics[width=.90\linewidth]{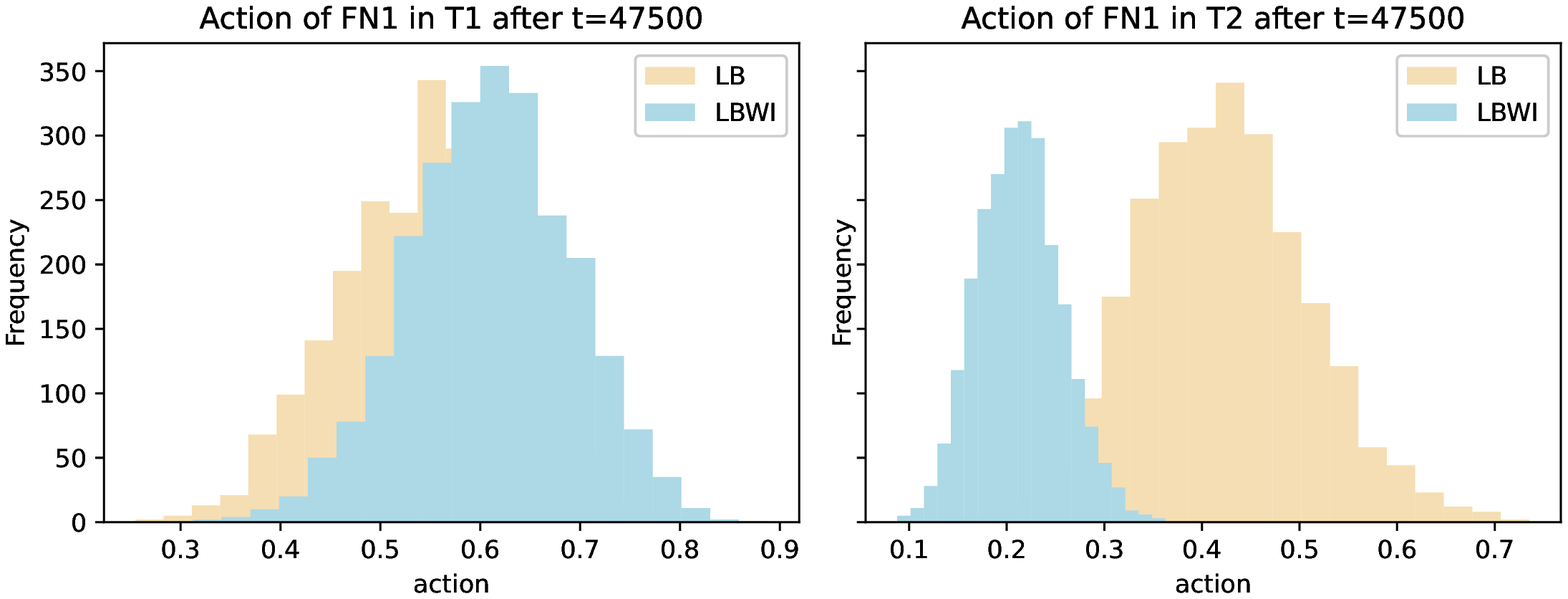}  
\end{subfigure} \newline
\vspace{-15pt}
\begin{subfigure}{.98\textwidth}
  \centering
  \hspace{-15pt}
  \includegraphics[width=.90\linewidth]{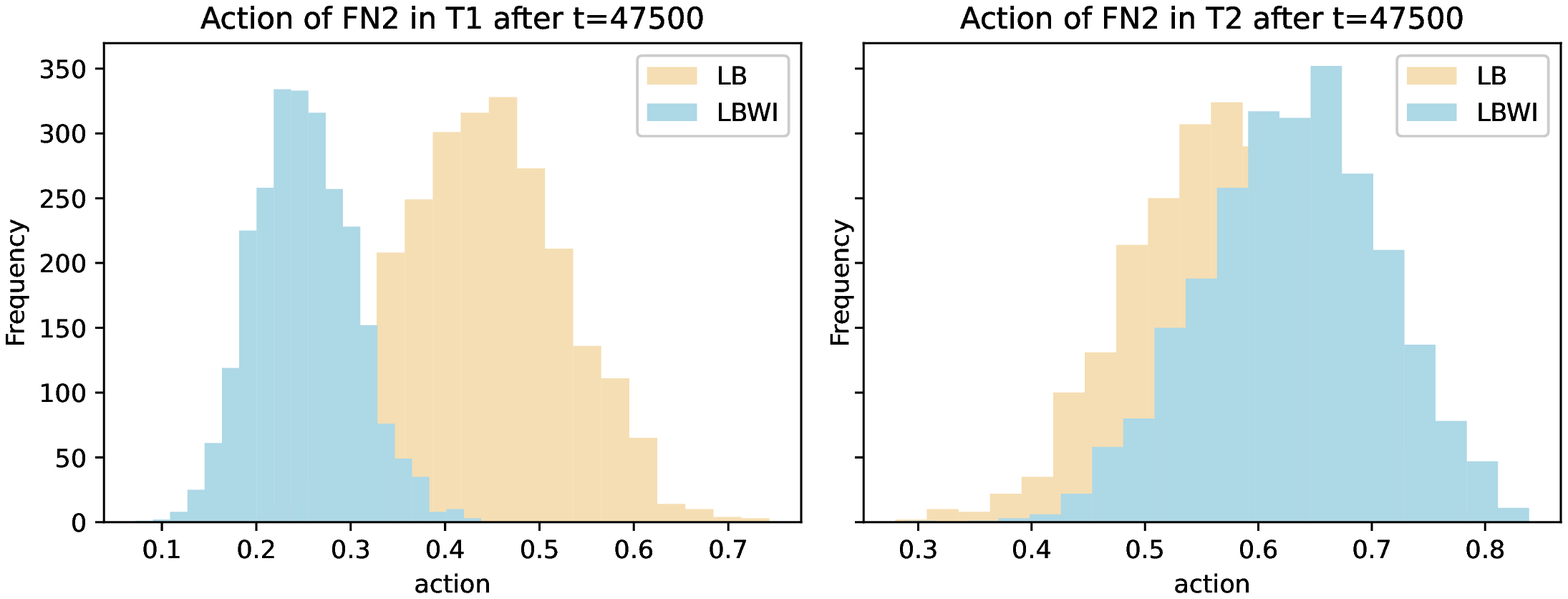}  
\end{subfigure}
\vspace{5pt}
\caption{Actions' frequency of two FNs \\ after $t = 45000$.}
\label{fig:e1lip}
\end{minipage}%
\begin{minipage}{.29\textwidth}
  \centering
  \includegraphics[width=0.9\linewidth]{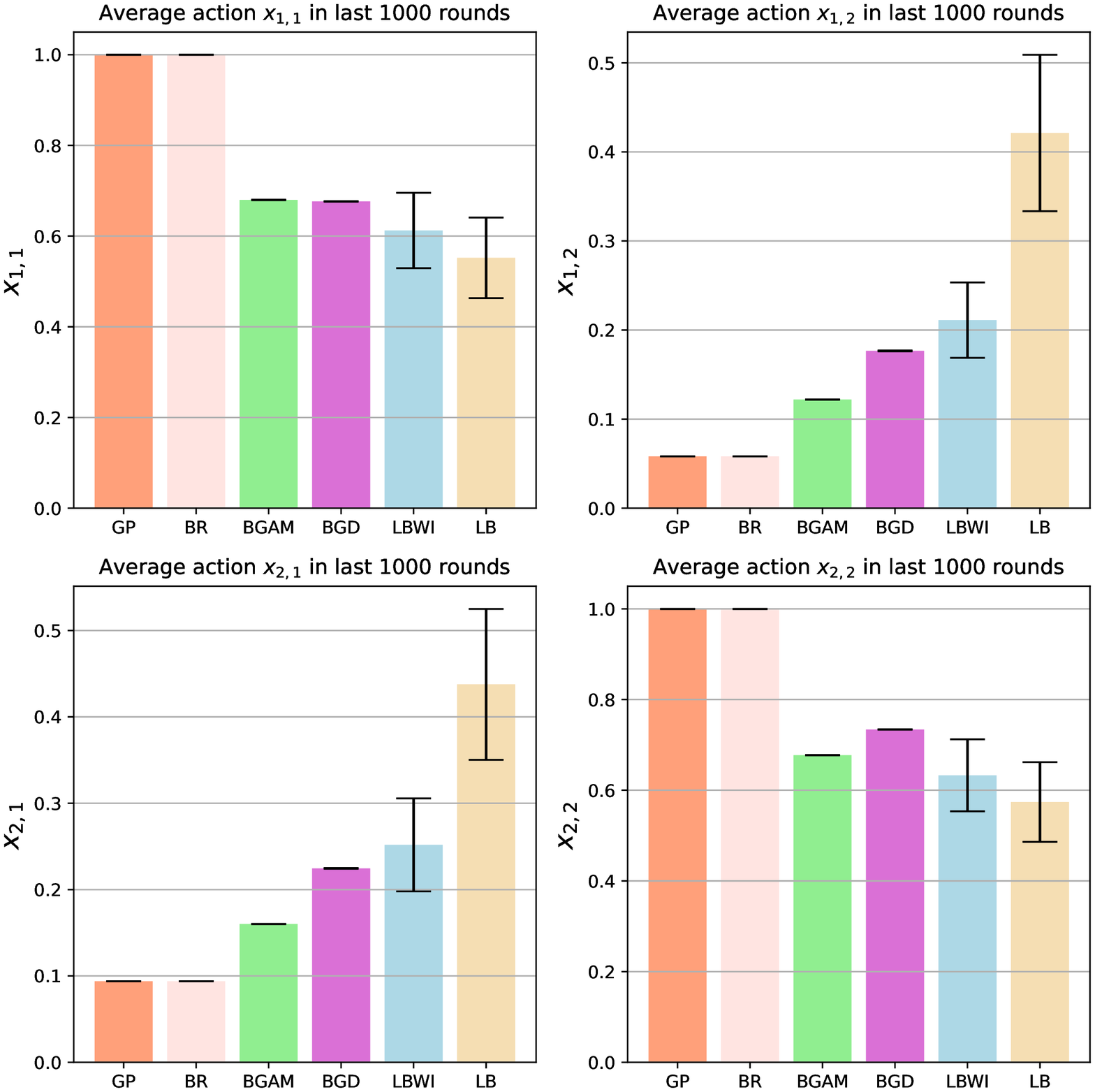}
  \caption{Average action in the final 1000 rounds.}
  \label{fig:necv}
\end{minipage}
\end{figure*}
\subsubsection{Results and Discussions} 
In this section, we compare the performance of BGAM and LBWI policies with their original versions bandit gradient descent (BGD) and Lipschitz bandit (LB), respectively. Note that BGD and the proposed BGAM algorithm belong to the class of bandit gradient optimization algorithms, while the Lipschitz EXP3 bandit with initialization, is a type of Lipschitz bandit strategy.

\textbf{Fig.~\ref{fig:fnact1}} shows the evolution of the average actions of each fog node in bandit gradient optimization algorithms. The results demonstrate that the average actions converge. Note that $x_{k,m}$, $k \in \mathcal{K}, m \in \mathcal{M}$ is the proportion that fog node (FN) $k$ prefers to contribute to task $m$, the real allocated amount is based on \eqref{eq:akm}. Besides, the convergence point is consistent with the conclusion that we drew from the utility plots. Both $x_{1,2}$ (fraction of Task 2 for Node 1) and $x_{2,1}$ (fraction of Task 1 for Node 2) converge towards some small value (almost zero) while the action of FN1 in Task 1 and FN2 in Task 2 keep increasing, indicating to have larger task share. 

Besides, we compare the evolution of actions of LBWI and LB. As described before, the first phase of the Lipschitz bandit strategies involves only uniform random selection; as such, we omit the selection frequency histograms for the initial phase. \textbf{Fig.~\ref{fig:e1lip}} demonstrates the selection frequency of different action fractions between $[0,1]$ by the LBWI and LB strategies after $90\%$ of rounds. We observe the same convergence tendency as \textbf{Fig.~\ref{fig:fnact1}}. After $90\%$ of rounds, the average for actions with LBWI strategy is $\bar{x}_{1,1}\approx 0.6$, $\bar{x}_{1,2}\approx 0.2$, $\bar{x}_{2,1}\approx 0.2$, and $\bar{x}_{2,2} \approx 0.6$. These results approve that LBWI is an effective strategy in the task allocation game. Compared to LB, LBWI performs better in identifying the suitable action fraction range, which means that initialization after the first phase enhances the performance of the Lipschitz bandit strategy. Nevertheless, compared to BGAM, the average action is not ideal, and the reason is that the Lipschitz bandit strategies sample the action fraction from the interval; thus, it is hard to achieve the optimal solution. Besides, in a lengthy game, the number of intervals is very large; that increases the computational cost and expands the exploration period. Consequently, the performance of BGAM is superior to that of LBWI. 

To validate the uniqueness of Nash equilibrium and to compare the performance of the proposed strategies, \textbf{Fig.~\ref{fig:necv}} shows the average actions within a small number of rounds near the end. We do the performance evaluation in different contexts, namely, online learning and game theory:
\begin{itemize}
    \item Greedy Projection (GP) \cite{zinkevich2003online}: This strategy stems from online convex optimization. It selects an arbitrary initial value $x_{k,m}^1$ and a sequence of learning rate $\eta^1, \eta^2, \ldots \in \mathbb{R}^+$. At time step $t$, it chooses the next vector $x_{k,m}^{t+1}$ as $$x_{k,m}^{t+1} = P(x_{k,m}^t + \eta^t \nabla u_{k,m}^t(x_{k,m}^t)).$$
    \item Best Response (BR) \cite{shoham2008multiagent}: This strategy is popular in game theory. Given $\boldsymbol{X}_{-k}$ as the strategies of all players excluding player $k$, then the player $k$'s best response strategy is $$\boldsymbol{x}_k^* = \text{argmax}_{\boldsymbol{x}_k}u_k(\boldsymbol{x}_k, \boldsymbol{X}_{-k}).$$
\end{itemize}
Both GP and BR strategies require full feedback. Given that, they converge swiftly. Besides, their convergence points are almost the same, i.e., they converge to the same Nash equilibrium. Finally, among four bandit-feedback strategies, the average actions of BGAM have the minimum distance to the optimal actions at Nash equilibrium. The average action profile of LBWI is comparable to that of BGD despite more deviation. \textbf{Fig.~\ref{fig:e1util}} shows the average utility of each node generated by different strategies. This figure confirms the results that appear in \textbf{Fig.~\ref{fig:necv}}. BGAM offers the highest utilities compared to other strategies with bandit feedback. Besides, both BGAM and LBWI perform better than their former version, BGD, and LB, respectively. In other words, our proposed modifications are remarkably effective, especially to solve the task allocation problem. The utility of LBWI is slightly less than BGD, which is due to its deviation. The utility of both fog nodes increases over time; However, unlike bandit gradient strategies, the Lipschitz bandit strategies suffer a low, non-increasing utility at the beginning, which arises due to the random selection in the first phase, and vanishes later. Also, in the second phase, the Lipschitz bandit strategies execute exploration in a vast space $\mathcal{O}(T^{\frac{1}{3}})$, which escalates the computational cost.  Therefore, especially in a long-run intensive sequential experiment, Lipschitz bandit strategies spend plenty of time and computational resources in exploration, which degrades their performance compared to the bandit gradient algorithms.

\begin{figure*}
\centering
\begin{minipage}{.24\textwidth}
  \vspace{-5pt}
  \centering
  \includegraphics[width=0.97\linewidth]{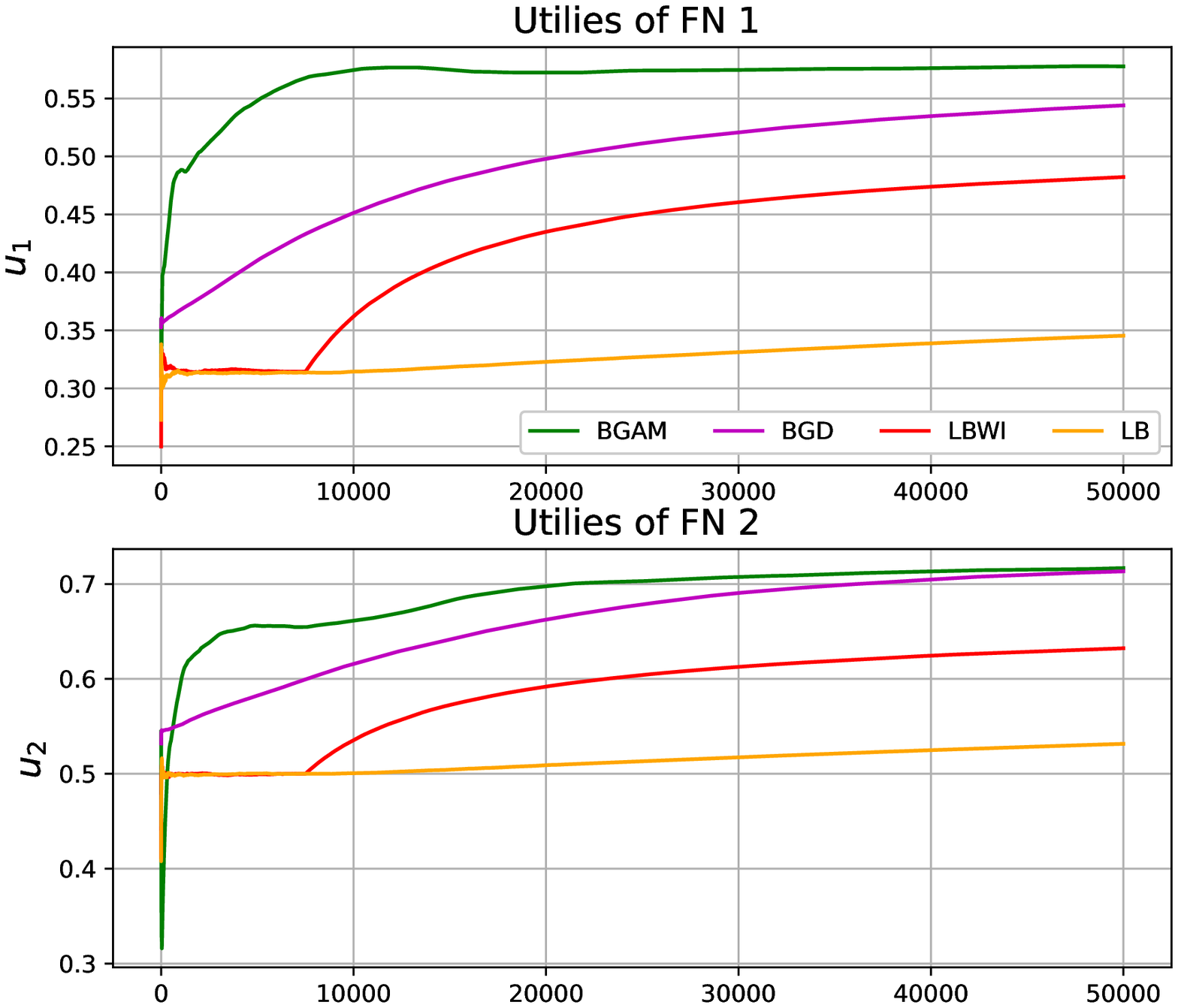}
  \caption{Utilities of two FNs.}
  \label{fig:e1util}
\end{minipage}
\begin{minipage}{.37\textwidth}
  \centering
  \includegraphics[width=0.95\linewidth]{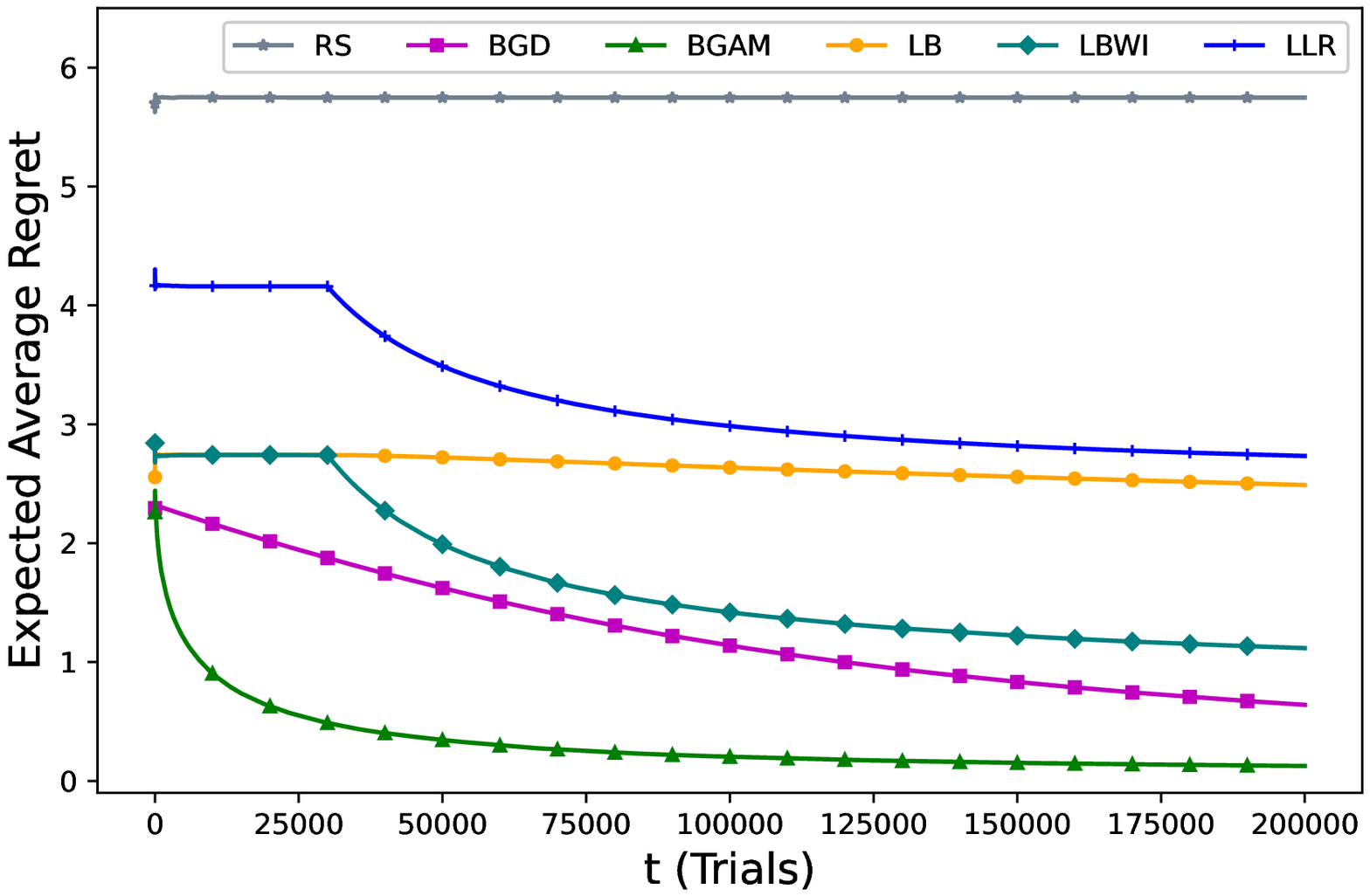}
  \caption{Expected average regret ($K=5,M=10$).}
  \label{fig:e2}
\end{minipage}%
\begin{minipage}{.37\textwidth}
  \centering
  \includegraphics[width=0.95\linewidth]{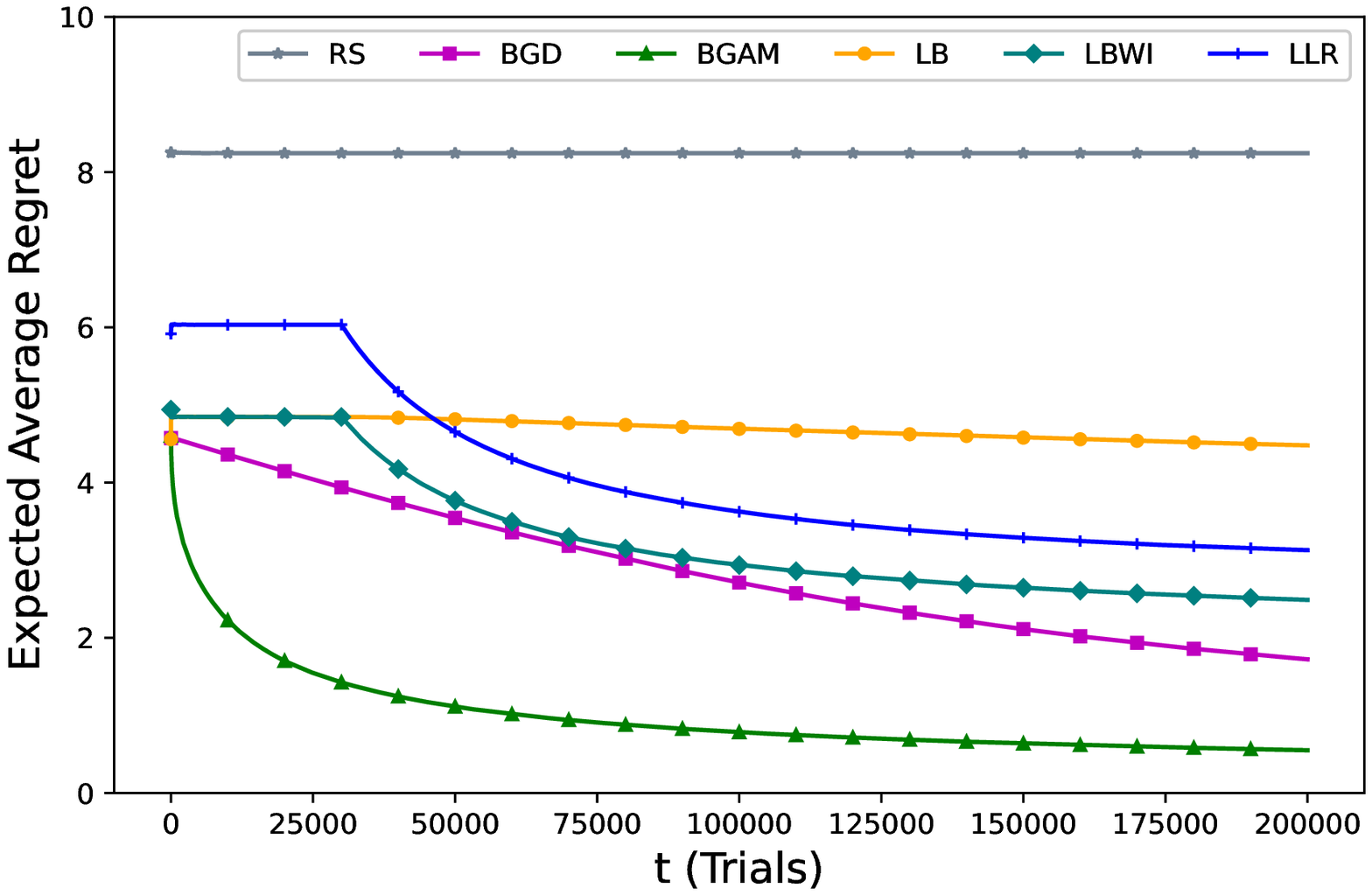}
  \caption{Expected average regret ($K=10,M=10$).}
  \label{fig:e3}
\end{minipage}
\end{figure*}
\subsection{Game II}
\label{sec:na-g2}
In this section, for a reliable performance evaluation, we perform some experiments with a larger network. Specifically, the network has five and ten fog nodes ($K = 5, 10$) that perform seven tasks ($M = 10$) cooperatively. The efficiency index and resource consumption index are obtained from a real-world dataset \textbf{Cloud-Fog Computing Dataset} \footnote{\url{https://www.kaggle.com/datasets/sachin26240/vehicularfogcomputing?resource=download}}. This dataset collects a detailed description of independent tasks and the information on 13 nodes (three cloud nodes and ten fog nodes). In this simulation, we only consider the performance in ten fog nodes and assume that the cloud nodes have unlimited computing ability. Details are available in Appendix~\ref{app:supl-e}.

As discussed in \textbf{Section~\ref{sec:na-g1}}, GP and BR dynamics reach the Nash equilibrium easily with the help of full-information feedback, which can be regarded as optimal strategies. Consequently, we use their obtained utility values at the convergence point as the maximum achievable utility to measure the regret of proposed policies. Besides the LLR strategy described in \textbf{Section~\ref{sec:llr}}, we use the uniformly random action selection (RS) as a benchmark. 

The results for $K = 5$, $M = 10$ is shown in \textbf{Fig.~\ref{fig:e2}} while \textbf{Fig.~\ref{fig:e3}} shows the expected average regret for $K = 10$, $M = 10$. All the experiments are based on ten-time average. Based on this figure, we conclude the followings:
\begin{itemize}
\item In both experiments, BGAM shows the best performance with the minimum regret (regret converges to zero) and the swiftest convergence. Compared with BGD, momentum plays a significant role in that achievement.
\item In the regret plot part of LBWI and LB are non-decreasing in the beginning. It is because of the pure random exploration in their first phase.  
\item The essential information required by the decision-making policies decreases in the following order: LLR, BGD(BGAM), LBWI, and LB. However, in \textbf{Fig.~\ref{fig:e2}} and \textbf{~\ref{fig:e3}}, LLR does not perform well, and the regret of LLR is even larger than LB in \textbf{Fig.~\ref{fig:e2}} and LBWI in \textbf{Fig.~\ref{fig:e3}}. The reason of poor LLR performance in the experiment with $K = 5$ fog nodes and $M = 10$ tasks is the following limitation: Every agent matches one and at most one task per round, and it is only implementable when $K \leq M$. Despite the centralized architecture, these limitations reduce the performance of the LLR and restrict its scope of application. 
\item Thanks to the momentum, BGAM converges faster than BGD and achieves much better performance than LLR even with less information. Besides, our proposed algorithms do not have the limitations of LLR, therefore, have more applications.
\item The regrets of LBWI and LB are higher than others due to the following reasons. Both methods have an initial phase of pure exploration in the action space, and for the LB, the number of intervals concluded from the first phase is always large ($>50$), which extends exploration. Because of our modifications, the initialization with information collected in the first phase is assistive in LBWI; hence, it incurs a lower regret than the LB strategy and achieves better performance than LLR even with less information.
\end{itemize}
\begin{remark}
According to \textbf{Fig.~\ref{fig:e2}} and \textbf{Fig.~\ref{fig:e3}}, only after $t=25000$ the average regret of BGAM decreases to less than a reasonable threshold ($<10\%$). That seems to be implausible in several applications; Nevertheless, as stated by \textbf{Theorem~\ref{the:eps-equ}}, one can suffice to a near-optimal solution, namely, $\epsilon$-Nash equilibrium, instead of achieving the unique Nash equilibrium. That shortens the convergence time significantly with negligible performance degradation. 
\end{remark}
\section{Conclusion} 
\label{sec:concl}
The task allocation problem in a fog computing infrastructure is studied in this paper. Every fog node has some preference over tasks, which is initially unknown. The self-interested fog nodes aim to perform the tasks in a distributed manner so that the outcome of consecutive interactions is an efficient equilibrium. We formulate the scenario as an online convex optimization problem.  We developed two decision-making policies for the formulated problem, namely BGAM and LBWI. Theoretically, we established that both methods guarantee a sublinear upper bound for the regret growth, and they converge to the unique Nash equilibrium. Numerical results showed that the proposed BGAM algorithm achieves a close-to-optimal utility performance superior to the existing algorithms. Future research directions include
studying the problem in even more realistic settings, for example, considering the constraints on fog nodes' task performance or distributed task assignments with delayed- or missed feedback in some rounds. In addition, it is also important to improve learning strategies to reduce the regret and complexity bound.
\section{Appendix}
\label{Sec:Appendix}
\subsection{Some Auxiliary Results Required for the Proofs}
\label{Sec:Auxiliary}
%
\begin{theorem} [\cite{fuente_2000,sundaram_1996}] 
\label{the:conv-conc}
Let $f: U \rightarrow \mathbb{R}$ be a twice continuously differentiable function defined on $U \subset \mathbb{R}^n$, i.e. $f\in C^2(U)$. Then 
\begin{itemize}
    \item[(i)] $f$ is convex iff the Hessian matrix $D^2f(x)$ is positive semidefinite for all $x\in U$, i.e.,
    \begin{gather*}
        \langle D^2f(x)h,h \rangle \geq 0, \ \text{for any} \ h \in \mathbb{R}^n.
    \end{gather*}
    If $\langle D^2f(x)h,h \rangle > 0, \ \text{for any} \ h \in \mathbb{R}^n \setminus \{0\}$, then $f$ is strictly convex.
    \item[(ii)] $f$ is concave iff the Hessian matrix $D^2f(x)$ is negative semidefinite for all $x\in U$, i.e.,
    \begin{gather*}
        \langle D^2f(x)h,h \rangle \leq 0, \ \text{for any} \ h \in \mathbb{R}^n.
    \end{gather*}
    If $\langle D^2f(x)h,h \rangle < 0, \ \text{for any} \ h \in \mathbb{R}^n \setminus \{0\}$, then $f$ is strictly concave.
\end{itemize}
\end{theorem}
\begin{lemma}[\cite{even2009convergence}]
\label{lem:cc}
Consider a socially concave game $\Gamma$. If for every player $k \in \mathcal{K}$ the utility function $u_{k}(\boldsymbol{x}_{k},\boldsymbol{X}_{-k})$ is twice differentiable in every $\boldsymbol{x}_{k} \in \mathcal{S}_{k}$, then $\Gamma$ is a concave game.
\end{lemma}
\begin{definition} [Diagonally Strictly Concave \cite{rosen1965existence}]
\label{def:dsc}
Define a weighted sum of utility functions $\sigma(\boldsymbol{X},\boldsymbol{r})= \sum_{k=1}^{K} r_{k} u_{k}(\boldsymbol{X})$, where $\boldsymbol{r} = \{r_{1}, \dots, r_{K}\}, r_{k} \geq 0, \ \forall k$. The pseudogradient of $\sigma(\boldsymbol{X},\boldsymbol{r})$ for any nonnegative $r$ is defined as
\begin{gather}
\nabla \sigma(\boldsymbol{X},\boldsymbol{r}) =  \left[ r_{1} \nabla_{1} u_{1}(\boldsymbol{X}), \cdots, r_{K}\nabla_{K} u_{K}(\boldsymbol{X}) \right]^T, 
\label{eq:pg}
\end{gather}
where $\nabla_{k}$ is the gradient with respect to $\boldsymbol{x}_{k}$. The function $\sigma(\boldsymbol{X},\boldsymbol{r})$ is called diagonally strictly concave (DSC) for a given non-negative $\boldsymbol{r}$ if for every distinct pair $\boldsymbol{X}^{0},\boldsymbol{X}^{1} \in \mathcal{S}$, we have
\begin{gather}
(\boldsymbol{X}^{1}-\boldsymbol{X}^{0})^{T}(\nabla \sigma(\boldsymbol{X}^1,\boldsymbol{r})-\nabla \sigma(\boldsymbol{X}^0,\boldsymbol{r})) < 0. 
\label{eq:dsc}
\end{gather}
\end{definition}
\begin{theorem}[\cite{rosen1965existence}]
\label{the:uniq-ne}
Consider a game with orthogonal constraint set. The function $\sigma(\boldsymbol{X},\boldsymbol{r})$ is diagonally strictly concave with positive definite $\boldsymbol{r}$. If a Nash equilibrium exists, it is unique.
\end{theorem}
\begin{theorem}[\cite{even2009convergence}] 
\label{the:eps-equ}
Consider a socially concave game $\Gamma$ with $K$ players. If every player $k$ plays according to a procedure with external regret bound $\mathcal{R}_k(t)$, then at time t, the followings hold:
\begin{itemize}
\item[(i)] The average strategy vector $\hat{\boldsymbol{X}}^t$ is an $\epsilon$- Nash equilibrium, where $\epsilon^t = \frac{1}{\lambda_{min}}\sum_{i\in \mathcal{K}}\frac{\lambda_i\mathcal{R}_i}{t}$ and $\lambda_{min} = \min_{i \in \mathcal{K}}\lambda_i$. 
\item[(ii)] The average utility of each player $k$ is close to her utility at $\hat{\boldsymbol{X}}^t$, the average vector of strategies. Formally,
\begin{gather*}
\abs{\hat{u}_k^t - u_k (\hat{\boldsymbol{X}}^t) } \leq \frac{1}{\lambda_k}\sum_{i\in \mathcal{K}}\frac{\lambda_i\mathcal{R}_i}{t}.
\end{gather*}
\end{itemize}
\end{theorem}
\begin{lemma}[\cite{even2009convergence}]
\label{Lm:SCConvergence}
Consider a socially concave game $\Gamma$ with $K$ players. If every player $k$ plays according to a procedure with non-external regret so that $\lim_{T \to \infty}\frac{1}{T}\mathcal{R}_T = 0$, then the players' joint action profile converges to a Nash equilibrium.
\end{lemma}
\begin{lemma}[\cite{lucas2018aggregated}]
\label{lem:vt-b}
Let $u^{1},u^{2},\cdots,u^{T}: \mathcal{S} \rightarrow \mathbb{R}$ be a sequence of concave and differentiable functions. Besides, $g^{1},g^{2},\cdots,g^{T}$ are the single-point estimation of gradient with $g^{t} = \nabla \hat{u}^{t}(y^{t})$ and $\norm{g^{t}} \leq G$. Also, $g_{1:T} = [g^1, \ldots, g^T]$, and $\beta$ is a constant. Select $\nu^{t}=\frac{\nu}{\sqrt{t}}$. Then we have
\begin{gather*}
\sum_{t=1}^T \nu^t (v^t)^2 = \sum_{t=1}^T \frac{\nu (v^t)^2}{\sqrt{t}} \leq \norm{g_{1:T}}_4^2 \sqrt{1+\log T} \frac{\nu}{(1-\beta)^2}.
\end{gather*}
\end{lemma}
\begin{definition}
\label{Def:Hes}
$H$ is a uniform bound for the Hessians of utility function $u$ if
\begin{gather*}
|\langle D^2u_k(\boldsymbol{x}_k)\boldsymbol{y},\boldsymbol{y} \rangle|\leq H \norm{\boldsymbol{y}}_{\infty}^2,
\end{gather*}
for any $\boldsymbol{x}_k \in [0,1]^M, \boldsymbol{y} \in [0,1]^M$.
\end{definition}
\begin{lemma}[\cite{bubeck2011lipschitz}]
\label{lem:apxl}
Let $N$ be the number of discritization intervals in the first phase of LBWI strategy. For $N \geq 3$, with probability at least $1-\frac{1}{T}$,
\begin{gather*}
L - \frac{7H}{N} \leq \tilde{L} \leq L + 2N\sqrt{\frac{2}{\mathcal{A}}\ln(2NT)},
\end{gather*}
where $H$ is defined in \textbf{Definition \ref{Def:Hes}}, $L$ is the Lipschitz constant, and $\tilde{L}$ is given by \eqref{eq:hatL}. Besides, $\mathcal{A}$ is the number of times each
arm is pulled independently.
\end{lemma}
\subsection{Proof of Proposition \ref{prop:lips}}
\label{Sec:Proposition1}
The gradient of utility function defined in (\ref{eq:util}) is 
\begin{gather}
\nabla u_k = \left[\nabla u_{k,1},\cdots, \nabla u_{k,M} \right]^T ,\label{eq:util-grad}
\end{gather}
where $\nabla u_{k,m} = \frac{X_{-k,m}}{(x_{k,m}+X_{-k,m})^2}\exp(-\frac{x_{k,m}}{\rho_{k,m}(x_{k,m}+X_{-k,m})}) + \epsilon_{k,m} - \kappa_{k,m}$ and $X_{-k,m} = \sum_{i\in \mathcal{K} \backslash {k}} x_{i,m}$. Because $\exp(-\frac{x_{k,m}}{\rho_{k,m}(x_{k,m}+X_{-k,m})}) \leq 1$ and $\frac{X_{-k,m}}{(x_{k,m}+X_{-k,m})^2} \leq \frac{1}{X_{-k,m}}$, there exists an $L$, such that $|\nabla u_{k,m}|<L$, $\forall x_{k,m} \in [0,1], k \in \mathcal{K}, m \in \mathcal{M}$.\footnote{The \textit{barrier to entry} in denominator of $a_{k,m}$ will avoid the extreme value of $L$.} That completes the proof.
\subsection{Proof of Proposition \ref{prop:conc-conv}}
\label{Sec:Proposition2}
According to \textbf{Theorem~\ref{the:conv-conc}}, at first we calculate the Hessian matrix of $u_k$ in $\boldsymbol{x}_k$ as
\begin{gather*}
D^2u_k(\boldsymbol{x_k}) = \left[ \begin{array}{cccc}
\frac{\partial^2u_k}{\partial x_{k,1}^2} & 0 & \ldots & 0 \\
\vdots & \vdots & \ddots & \vdots \\
0 & 0 & \ldots & \frac{\partial^2u_k}{\partial x_{k,M}^2}
\end{array}\right]
\end{gather*}
where $\frac{\partial^2u_k}{\partial x_{k,m}^2} = -e^{\frac{x_{k,m}}{\rho_{k,m}(x_{k,m}+X_{-k,m})}}(\frac{2X_{-k,m}}{(x_{k,m}+X_{-k,m})^3} + \frac{X_{-k,m}^2}{\rho_{k,m}(x_{k,m}+X_{-k,m})^4})$ is non-positive $\forall~m \in \mathcal{M}$. Thus for any $h \in \mathbb{R}^M$, it satisfies
$\langle D^2u_k(\boldsymbol{x}_k), h \rangle \leq 0$, i.e., $u_k(\boldsymbol{x}_k,\boldsymbol{X}_{-k})$ is concave in $\boldsymbol{x}_k$.

Similarly, the Hessian matrix of $u_k$ in $\boldsymbol{X}_{-k}$ is
\begin{gather*}
D^2u_k(\boldsymbol{X}_{-k}) = \left[ \begin{array}{cccc}
\frac{\partial^2u_k}{\partial X_{-k,1}^2} & 0 & \ldots & 0 \\
\vdots & \vdots & \ddots & \vdots \\
0 & 0 & \ldots & \frac{\partial^2u_k}{\partial X_{-k,M}^2}
\end{array}\right]
\end{gather*}
where $\frac{\partial^2 \varphi}{\partial X_{-k,m}^2} =  e^{-\frac{x_{k,m}}{\rho_{k,m}(x_{k,m}+X_{-k,m})}}\frac{2x_{k,m} X_{-k,m}+(2\rho_{k,m} - 1)x_{k,m}^2}{\rho_{k,m}(x_{k,m}+X_{-k,m})^4}$ is non-negative $\forall~m\in \mathcal{M}$ by assuming $\rho_{\min} > 0.5$. Thus for any $h\in \mathbb{R}^M$, it satisfies $
\langle D^2u_k(\boldsymbol{X}_{-k}), h \rangle \geq 0$, i.e., $u_k$ is convex in $\boldsymbol{X}_{-k}$. That completes the proof.
\subsection{Proof of Proposition \ref{prop:ta-sc}}
\label{Sec:Proposition3}
By \textbf{Proposition~\ref{prop:conc-conv}}, the utility function $u_k$ is concave in $\boldsymbol{x}_k$, the proposal of fog node $k$, and convex in the actions of other fog nodes $\boldsymbol{X}_{-k}$, which satisfies the two conditions in \textbf{Definition~\ref{def:sc-concv}}. Therefore, the task allocation game $\Gamma$ is a socially concave game. Besides, the utility function $u_k$ is twice differentiable and therefore, by \textbf{Lemma \ref{lem:cc}}, $\Gamma$ is also a concave game (\textbf{Definition \ref{def:concv}}). 
\subsection{Proof of Proposition \ref{prop:ta-ue}}
\label{Sec:Proposition4}
Define a weighted sum of the utility functions in task allocation game $\Gamma$ as 
\begin{gather*}
\sigma(\boldsymbol{X}, \boldsymbol{r}) = \sum_{k=1}^K r_k u_k(\boldsymbol{X})= \sum_{k=1}^K r_k u_k(\boldsymbol{x}_k,\boldsymbol{X}_{-k}),
\end{gather*}
where $\forall k, r_k > 0$. The pseudogradient of $\sigma(\boldsymbol{X}, \boldsymbol{r})$ is defined by \eqref{eq:pg}, where $\nabla_k u_k(\boldsymbol{X}) = \nabla u_k(\boldsymbol{x}_k)$, which is  given by (\ref{eq:util-grad}) with 
\begin{align*}
    \nabla u_{k,m} = \frac{X_{-k,m}}{(x_{k,m}+X_{-k,m})^2}&\exp(-\frac{x_{k,m}}{\rho_{k,m}(x_{k,m}+X_{-k,m})}) \\
    & + \epsilon_{k,m} - \kappa_{k,m}. 
\end{align*}
To prove the uniqueness of Nash equilibrium, we need to prove the inequality~(\ref{eq:dsc}), which is equivalent to the following inequality in task allocation game $\Gamma$:
\begin{gather}
\sum_{k} \sum_{m} r_k(x_{k,m}^1-x_{k,m}^0)(\nabla u_{k,m}(x_{k,m}^1)-u_{k,m}(x_{k,m}^0)) < 0. \notag
\end{gather}
As $\frac{\partial^2 u_{k,m}}{\partial x_{k,m}^2} = -e^{-\frac{x_{k,m}}{\rho_{k,m}(x_{k,m}+X_{-k,m})}}(\frac{2X_{-k,m}}{(x_{k,m}+X_{-k,m})^3} + \frac{(X_{-k,m})^2}{\rho_{k,m}(x_{k,m}+X_{-k,m})^4}) < 0$, $\forall x_{k,m} \in [0,1]$. Thus, by the mean value theorem, there exists $x'_{k,m} \in [x_{k,m}^0,x_{k,m}^1]$ for any $x_{k,m}^0 \neq x_{k,m}^1$ so that
\begin{gather*}
\frac{\nabla u_{k,m}(x_{k,m}^1)-\nabla u_{k,m}(x_{k,m}^0)}{x_{k,m}^1-x_{k,m}^0} 
= \frac{\partial^2 u_{k,m}}{\partial x_{k,m}^{'2}} < 0.
\end{gather*}
That is then equivalent to 
\begin{gather}
r_k(x_{k,m}^1-x_{k,m}^0)(\nabla u_{k,m}(x_{k,m}^1) - \nabla u_{k,m}(x_{k,m}^0))< 0,  \notag
\end{gather}
for all $r_k>0$. Summing up with $k \in \mathcal{K}, m \in \mathcal{M}$, we conclude
\begin{gather}
\sum_{k} \sum_{m} r_k (x_{k,m}^1-x_{k,m}^0)(\nabla u_{k,m}(x_{k,m}^1)-\nabla u_{k,m}(x_{k,m}^0)) < 0. \notag
\end{gather}
According to \textbf{Definition~\ref{def:dsc}}, $\sigma(\boldsymbol{x}, \boldsymbol{r})$ is diagonally strictly concave. Therefore, by \textbf{Theorem \ref{the:uniq-ne}}, Nash equilibrium of task allocation game $\Gamma$ is unique. That completes the proof.
\subsection{Proof of Proposition \ref{pr:RegretBGAM}}
\label{Sec:Proposition5}
Because $u^t$ is the same as the utility function with the shifted action space, the regret
\begin{gather}
R_{k,m}^T(x_{k,m}) = \max_{x_{k,m}}~\sum_{t=1}^T u_{k,m}^t(x_{k,m}) - \sum_{t=1}^Tu_{k,m}^t(x_{k,m}^t) \notag
\end{gather}
equals to
\begin{gather}
R^T(z) = \max_z~\sum_{t=1}^T u^t(z) - \sum_{t=1}^T u^t(z^t). \label{eq:rgt-bgam}
\end{gather}
Thus we use the regret $R^T(z)$ in (\ref{eq:rgt-bgam}). By (\ref{eq:est-z}) we have \cite{flaxman2004online} $\hat{u}^t(y) = \mathbb{E}[u^t(z^t)]$, which results
\begin{gather}
\max_z~\mathbb{E}[\sum_{t=1}^T u^t(z)]=\max_{(y+\sigma c)}~\mathbb{E}[\sum_{t=1}^T u^t(y+\sigma c)]=\max_y \sum_{t=1}^T \hat{u}^t(y) \notag
\end{gather}
Then, the expectation of regret $R^T(z)$ will be 
\begin{gather}
\mathbb{E}[R^T(z)] = \max_y \sum_{t=1}^T \hat{u}^t(y) - \sum_{t=1}^T \hat{u}^t(y^t).
\end{gather}
Define $y^*$ as the optimal $y$ that maximize $\hat{u}^t(y)$. We can bound the expected difference between $\hat{u}^t(y^t)$ and $\hat{u}^t(y^*)$ in terms of gradient with \cite{flaxman2004online}
\begin{gather}
\mathbb{E}[\hat{u}^t(y^*) - \hat{u}^t(y^t)] \leq \mathbb{E}[g^t(y^* - y^t)]. 
\label{eq:ug}
\end{gather}
From the update rules (\ref{eq:v}) and (\ref{eq:y}), we can conclude that
\begin{gather*}
y^{t+1} = P_{(1-\alpha)S}(y^t +\nu^t (\beta v^{t-1}+g^t)).
\end{gather*}
By Zinkevich and Flaxman's analysis \cite{zinkevich2003online,flaxman2004online}, we have $\norm{y^* - P_s(y)} \leq  \norm{y^* - y}$. Hence, 
\begin{align*}
\norm{y^* - y^{t+1}}^2 
&\leq \norm{y^* - y^t}^2 -2\nu^t (\beta v^{t-1}+g^t)(y^* -y^t) \notag \\
& \quad \quad + (\nu^t v^t)^2.
\end{align*}
Rearranging the expression, we get
\begin{align}
g^t(y^*-y^t) \leq  & \frac{1}{2\nu^t} (\norm{y^* - y^t}^2 - \norm{y^* - y^{t+1}}^2) \notag \\
& \quad -\beta v^{t-1}(y^* -y^t) + \frac{\nu^t}{2} (v^t)^2. 
\label{eq:ine-gt}
\end{align}
By summing inequality~\eqref{eq:ine-gt} and combining it with the inequality~(\ref{eq:ug}), we arrive at 
\begin{align*}
&\sum_{t=1}^T \mathbb{E}[\hat{u}^t(y^*) - \hat{u}^t(y^t)] \leq \sum_{t=1}^T  \mathbb{E}[g^t(y^*-y^t)] \\
\leq & \frac{D^2}{2\nu^1} + \frac{D^2}{2} (\frac{1}{\nu^t} - \frac{1}{\nu}) + \sum_{t=1}^T \frac{\nu^t}{2} (v^t)^2 + \sum_{t=1}^T D\beta v^t \\
\leq & \frac{D^2(1+\sqrt{T})}{2\nu} + \sum_{t=1}^T \frac{\nu^t}{2} (v^t)^2 + \sum_{t=1}^T D\beta v^t.
\end{align*}
The second inequality follow from the bounding assumptions $\norm{y^i - y^j}_2 \leq D$. From (\ref{eq:egrad}), it follows that
\begin{gather*}
\norm{g^t} = \norm{\frac{1}{\sigma} u^t(y^t + \sigma c^t)c^t} \leq \frac{U}{\sigma},
\end{gather*}
where $U$ is the maximum value of utility function. According to \eqref{eq:v}, we can have $v^t \leq \frac{U(1-\beta^{t-1})}{1-\beta}$, thus,
\begin{gather*}
    D\beta\sum_{t=1}^T v^t \leq \frac{D \beta U (1 - \beta^T)}{(1-\beta)^2} \leq \frac{D \beta U}{(1-\beta)^2}. 
\end{gather*}
According to \textbf{Lemma~\ref{lem:vt-b}}, we can conclude
\begin{align}
\sum_{t=1}^T \mathbb{E}[\hat{u}^t(y^*) - \hat{u}^t(y^t)] & \leq  \frac{\nu \sqrt{1+ \log T}}{2} \norm{g_{1:T}}_4^2 \frac{1}{(1-\beta)^2} \notag \\
& \quad +\frac{D \beta U}{(1-\beta)^2}+\frac{D^2(1+\sqrt{T})}{2\nu}. 
\label{eq:regt-y}
\end{align}
According to references \cite{lucas2018aggregated,duchi2011adaptive}, $\norm{g_{1:T}}_4^2 \leq \frac{U}{\sigma} \sqrt{T}$. Therefore, we have
\begin{align}
&\max_{y \in (1-\alpha)S}\sum_{t=1}^T\hat{u}^t(y)  - \mathbb{E}[\sum_{t=1}^T \hat{u}^t(y^t)] \notag \\ & \leq \frac{D^2(1+\sqrt{T})}{2\nu} +(D\beta U+ \frac{\nu U \sqrt{T}\sqrt{1+ \log T}}{2\sigma})\frac{1}{(1-\beta)^2}. 
\label{eq:regt-yT}
\end{align}
In \textbf{Proposition~\ref{prop:lips}}, the utility function has a Lipschitz constant $L$. Hence \cite{flaxman2004online}, $|\hat{u}^t(y^t) - u^t(y^t)| \leq \sigma L$, $|\hat{u}^t(y^t) - u^t(z^t)| \leq  2\sigma L$. These imply
\begin{align}
\max_{z \in (1-\alpha)S}\sum_{t=1}^T&(u^t(z)-\sigma L)  - \mathbb{E}[\sum_{t=1}^T (u^t(z^t)+2\sigma L)] \notag \\
\leq  \frac{D^2(1+\sqrt{T})}{2\nu} &+ (D\beta U+ \frac{\nu U \sqrt{T}\sqrt{1+ \log T}}{2\sigma})\frac{1}{(1-\beta)^2}, \notag
\end{align}
\begin{align}
&\max_{z \in (1-\alpha)S}~\sum_{t=1}^T u^t(z) - \mathbb{E}[\sum_{t=1}^T u^t(z^t)] \leq \frac{D^2(1+\sqrt{T})}{2\nu} \notag  \\
& +3\sigma LT  + (D\beta U+ \frac{\nu U \sqrt{T}\sqrt{1+ \log T}}{2\sigma})\frac{1}{(1-\beta)^2}. 
\label{eq:rgt-1}
\end{align}
Because $u^t$ is concave in $z^t$ and $0 \in S$,
\begin{align*}
\max_{z \in (1-\alpha)S}~u^t(z) &= \max_{z \in S}~ u^t((1-\alpha)z) \notag \\
&\geq \max_{z \in S}~(\alpha u^t(0) + (1-\alpha)u^t(z)) \notag \\
& \geq \max_{z \in S}~(\alpha (u^t(0)-u^t(z)) + u^t(z)) \notag \\
&\geq \max_{z \in S}~(-2\alpha U + u^t(z)).
\end{align*}
With the above inequality (\ref{eq:rgt-1}) we have
\begin{align*}
 &\max_{z \in S}~\sum_{t=1}^T u^t(z)-\mathbb{E}[\sum_{t=1}^T u^t(z^t)]  \leq 3\sigma LT + 2 \alpha UT  \\
 & \quad + \frac{D^2(1+\sqrt{T})}{2\nu}+(D\beta U+ \frac{\nu U \sqrt{T}\sqrt{1+ \log T}}{2\sigma})\frac{1}{(1-\beta)^2} 
\end{align*}
Select $\sigma = T^{-0.25}\sqrt{\frac{RUr}{3(Lr+U)}}, \alpha = \frac{\sigma}{r}$, the expected regret is bounded by
\begin{align}
\mathbb{E}[R(T)]  \leq  & \frac{D^2(1+\sqrt{T})}{2\nu} +2 \frac{U}{r}T^{\frac{3}{4}}\sqrt{\frac{RUr}{3(Lr+U)}}\notag \\
& + \frac{\nu U T^{\frac{3}{4}} \sqrt{1+ \log T}}{2\sqrt{\frac{RUr}{3(Lr+U)}}} \frac{1}{(1-\beta)^2} + \frac{D\beta U}{(1-\beta)^2}
\end{align}
%
\subsection{Proof of Proposition \ref{pro:Lip}}
\label{Sec:Proposition6}
According to \textbf{Proposition~\ref{prop:lips}}, the utility function of fog node $k$'s action satisfies the Lipschitz condition. Therefore, fog node $k$ with action space $\mathcal{S}_k \in [0,1]^M$ and utility function $u_k$ in task allocation game $\Gamma$ can be modelled as a Lipschitz bandit. That completes the proof.
\subsection{Proof of Proposition \ref{pr:RegretLBWI}}
\label{Sec:Proposition7}
Considering the average utility in $N$ discretized bins, the average value of $u_{k,m}$ indexed by $n=0,1,\ldots,N-1$ is
\begin{gather*}
\bar{u}_{k,m}[n] = \int_{\frac{n}{N}}^{\frac{n+1}{N}} u_{k,m}(x)dx.
\end{gather*}
Because the utility function is L-Lipschitz, we have \cite{bubeck2011lipschitz}
\begin{gather*}
\max_{x \in [0,1]}~u_{k,m}(x)-\max_n~\bar{u}_{k,m}(n) \leq \frac{L}{N}.
\end{gather*}
Therefore, the expected regret bound of exploration-exploitation strategy with $N$-dimensional discrete action space yields
\begin{gather}
\mathbb{E}[R(T)] \leq T\frac{L}{N} + \mathcal{R}(T,N),
\end{gather}
where $\mathcal{R}(T,N)$ is a function that depends on round number $T$ and Lipschitz constant $L$ to represent the regret bound under discrete condition. By substitute the parameters as given in our algorithm we conclude \cite{bubeck2011lipschitz},
\begin{align}
\mathbb{E}[R(T)] &\leq T_1 + \mathbb{E}[\frac{LT}{\tilde{N}} + \mathcal{R}(T-T_1,\tilde{N})] \notag \\
&= \mathcal{A}N + \mathbb{E}[\frac{LT}{\tilde{N}} + \mathcal{R}(T-\mathcal{A}N,\tilde{N})],
\end{align}
where $\mathcal{A}N$ is the regret caused by random selection in phase I and $\mathcal{R}(T-\mathcal{A}N,\tilde{N})$ is the regret generated by the EXP3 strategy in phase II. The initialization of the weight matrix in EXP3 does not influence the regret bound and $\mathcal{R}(T',N') = 2.63\sqrt{TN'\ln N'}$ according to \cite{auer2002nonstochastic}. With 
\begin{align}
\tilde{N} &= N \left \lceil{\frac{\tilde{L}^{\frac{2}{3}}T^{\frac{1}{3}}}{N}}\right \rceil  \leq  \tilde{L}^{\frac{2}{3}}T^{\frac{1}{3}}(1+\frac{N}{\tilde{L}^{\frac{2}{3}}T^{\frac{1}{3}}}),
\end{align}
and by substituting the parameters in the bound $\mathcal{R}(T',N') = 2.63\sqrt{TN'\ln N'}$, we arrive at 
\begin{align}
\mathcal{R}(T,\tilde{N}) &\leq r\sqrt{\tilde{L}^{\frac{2}{3}}T^{\frac{4}{3}}(1+\frac{N}{\tilde{L}^{\frac{2}{3}}T^{\frac{1}{3}}})\ln(\tilde{L}^{\frac{2}{3}}T^{\frac{1}{3}}(1+\frac{N}{\tilde{L}^{\frac{2}{3}}T^{\frac{1}{3}}}))} \notag \\
&\leq r\sqrt{\tilde{L}^{\frac{2}{3}}T^{\frac{4}{3}}e\ln(e\tilde{L}^{\frac{2}{3}}T^{\frac{1}{3}})} \leq r'\tilde{L}^{\frac{2}{3}}T^{\frac{5}{6}}.
\end{align}
The first inequality follows from $(1+\frac{N}{\tilde{L}^{\frac{2}{3}}T^{\frac{1}{3}}})< e$ whenever $N <(e-1)\tilde{L}^{\frac{2}{3}}T^{\frac{1}{3}}$. The second inequality is concluded from the inequality $\ln(ex) = 1+ \ln x \leq x$ whenever $x>0$. Besides, we have $r = 2.63$ and $r' = 2.63\sqrt{e} \approx 4.34$.
According to \textbf{Lemma~\ref{lem:apxl}} and inequality $(x_1 + \cdots + x_p)^r \leq x_1^r + \cdots + x_p^r$ \cite{bubeck2011lipschitz}, we can conclude the following: With probability at least $1-\frac{1}{T}$, it holds
\begin{align}
\mathcal{R}(T,\tilde{N}) &\leq r'\tilde{L}^{\frac{2}{3}}T^{\frac{5}{6}} \leq r'[L + 2N\sqrt{\frac{2}{\mathcal{A}}\ln(2NT)})]^{\frac{2}{3}}T^{\frac{5}{6}} \notag \\
&\leq r'T^{\frac{5}{6}}(L^{\frac{2}{3}} + (2N\sqrt{\frac{2}{\mathcal{A}}\ln(2NT)})^{\frac{2}{3}}).
\end{align}
Moreover, with probability at least $1-\frac{1}{T}$, we have
%
$\tilde{L} \geq L - \frac{7H}{N} \geq \frac{L}{8}$.
%
If $N \geq \frac{8H}{L}$, where $H$ is defined in \textbf{Definition \ref{Def:Hes}}. Therefore,
\begin{align*}
\frac{LT}{\tilde{N}} \leq \frac{LT}{\tilde{L}^{\frac{2}{3}}T^{\frac{1}{3}}} \leq \frac{LT^{\frac{2}{3}}}{\frac{L}{8}^{2/3}} = 4L^{\frac{1}{3}}T^{\frac{2}{3}}.
\end{align*}
Putting things together, we get
\begin{align}
&\mathbb{E}[R(T)] \leq T_1 + \mathbb{E}[\frac{LT}{\tilde{N}} + \mathcal{A}(T-T_1,\tilde{N})] \notag \\
&\leq \mathcal{A}N + 4L^{\frac{1}{3}}T^{\frac{2}{3}} + r'T^{\frac{5}{6}}(L^{\frac{2}{3}} + (2N\sqrt{\frac{2}{\bar{\mathcal{A}}}\ln(2NT)})^{\frac{2}{3}}),
\end{align}
which can be concluded as the following: With probability at least $1-\frac{1}{T}$,
%
$\mathbb{E}[R(T)] \leq \tilde{\mathcal{O}}(T^{\frac{5}{6}})$.
%
\subsection{Proof of Proposition \ref{pr:Convergence}}
\label{Sec:Proposition8}
By \textbf{Proposition~\ref{prop:ta-sc}}, the task allocation game is a socially concave game. By \textbf{Proposition~\ref{prop:ta-ue}}, if the game converges to a Nash equilibrium, then that equilibrium is the unique equilibrium of the game. Thus, by \textbf{Theorem \ref{Lm:SCConvergence}}, to achieve equilibrium for every fog node $k \in \mathcal{K}$, it suffices that each fog node $k$ plays according to a no-regret decision-making strategy. By \textbf{Proposition~\ref{pr:RegretBGAM}} and \textbf{Proposition~\ref{pr:RegretLBWI}}, BGAM and LBWI are no-regret policies. 
\subsection{Supplementary of Simulation}
\label{app:supl-e}
In the \textbf{Cloud-Fog Computing Dataset}, we characterize the fog nodes' efficiency indices by considering the CPU rate and usage information with respect to each type of task. For cost indices, we use the CPU, memory, and bandwidth usage information with respect to each task type. Besides, all the indices are normalized to be in $[0,1]$. Since the power consumption is unavailable in the dataset, we assume the indices are far less than the other two. The following matrices show the value of indices ($K = 10$, $M = 10$).
\begin{gather*}
\tiny
\boldsymbol{\rho} =\begin{bmatrix}
0.47 & 0.40 & 0.65 & 0.30 & 0.34 & 0.51 & 0.27 &  0.44 & 0.57 & 0.57 \\
0.68 & 0.58 & 0.78 & 0.38 & 0.28 & 0.45 & 0.19 & 0.50 & 0.65 & 0.99 \\
0.98 & 0.70 & 0.18 & 0.99 & 0.56 & 0.61 & 0.68 & 0.57 & 0.51 & 0.33 \\
0.47 & 0.56 & 0.45 & 0.49 & 0.26 & 0.23 & 0.52 & 0.67 & 0.60 & 0.97 \\
0.30 & 0.37 & 0.33 & 0.23 & 0.74 & 0.39 & 0.71 & 0.48 & 0.48 & 0.83 \\
0.99 & 0.22 & 0.59 & 0.45 & 0.54 & 0.32 & 0.63 & 0.75 & 0.83 & 0.50\\
0.25 & 0.36 & 0.70 & 0.51 & 0.99 & 0.99 & 0.78 & 0.74 & 0.65 & 0.31\\
0.81 & 0.99 & 0.66 & 0.33 & 0.79 & 0.27 & 0.99 & 0.56 & 0.49 & 0.79\\
0.25 & 0.83 & 0.99 & 0.76 & 0.80 & 0.56 & 0.30 & 0.67 & 0.48 & 0.45 \\
0.44 & 0.50 & 0.47 & 0.52 & 0.70 & 0.72 & 0.50 & 0.82 & 0.39 & 0.91 
\end{bmatrix},
\end{gather*}
%
\begin{gather*}
\hspace{-15pt}
\tiny
\boldsymbol{\epsilon} =\begin{bmatrix}
0.014 & 0.006 & 0.004 & 0.008 & 0.008 & 0.012 & 0.02  & 0.006 & 0.006 & 0.04 \\
0.012 & 0.004 & 0.02  & 0.004 & 0.008 & 0.016 & 0.006 & 0.02  & 0.016 & 0.018 \\
0.008 & 0.042 & 0.008 & 0.02  & 0.018 & 0.04  & 0.008 & 0.01  & 0.006 & 0.002 \\
0.026 & 0.016 & 0.01  & 0.012 & 0.06  & 0.004 & 0.02  & 0.016 & 0.014 & 0.018 \\
0.018 & 0.046 & 0.014 & 0.028 & 0.046 & 0.004 & 0.024 & 0.022 & 0.024 & 0.01 \\
0.02  & 0.012 & 0.024 & 0.016 & 0.014 & 0.006 & 0.02  & 0.006 & 0.026 & 0.04 \\
0.012 & 0.004 & 0.02  & 0.004 & 0.008 & 0.036 & 0.006 & 0.02  & 0.016 & 0.018 \\
0.008 & 0.01  & 0.008 & 0.02  & 0.018 & 0.032 & 0.008 & 0.01  & 0.026 & 0.016 \\
0.026 & 0.036 & 0.01  & 0.012 & 0.06  & 0.018 & 0.02  & 0.016 & 0.014 & 0.024 \\
0.018 & 0.006 & 0.014 & 0.02  & 0.046 & 0.004 & 0.024 & 0.018 & 0.006 & 0.01 
\end{bmatrix},
\end{gather*}
\begin{gather*}
\tiny
\boldsymbol{\kappa} = \begin{bmatrix}
0.27 & 0.36 & 0.27 & 0.42 & 0.24 & 0.26 & 0.40 & 0.43 & 0.12 & 0.20 \\
0.16 & 0.42 & 0.19 & 0.42 & 0.40 & 0.28 & 0.38 & 0.16 & 0.28 & 0.20 \\
0.14 & 0.24 & 0.40 & 0.16 & 0.39 & 0.28 & 0.24 & 0.35 & 0.23 & 0.41 \\
0.33 & 0.40 & 0.40 & 0.22 & 0.39 & 0.42 & 0.17 & 0.02 & 0.40 & 0.19 \\
0.39 & 0.43 & 0.37 & 0.44 & 0.26 & 0.42 & 0.24 & 0.41 & 0.12 & 0.09 \\
0.18 & 0.39 & 0.19 & 0.19 & 0.29 & 0.42 & 0.16 & 0.32 & 0.11 & 0.21 \\
0.29 & 0.28 & 0.24 & 0.26 & 0.21 & 0.18 & 0.18 & 0.43 & 0.32 & 0.20 \\
0.15 & 0.24 & 0.20 & 0.24 & 0.19 & 0.42 & 0.13 & 0.29 & 0.11 & 0.37 \\
0.30 & 0.18 & 0.17 & 0.19 & 0.22 & 0.26 & 0.34 & 0.32 & 0.19 & 0.21 \\
0.30 & 0.40 & 0.27 & 0.32 & 0.32 & 0.25 & 0.28 & 0.14 & 0.48 & 0.09
\end{bmatrix}. 
\end{gather*}
\ifCLASSOPTIONcaptionsoff
  \newpage
\fi
\bibliographystyle{IEEEtran}
\bibliography{ref}
\end{document}